\documentclass[acmsmall,screen,nonacm]{acmart}

\usepackage{amsmath}
\usepackage{mathtools}
\usepackage{stmaryrd}
\usepackage{booktabs}
\usepackage{microtype}
\usepackage{xspace}
\usepackage{enumitem}
\usepackage{tikz}
\usetikzlibrary{decorations.pathreplacing}
\usepackage{listings}
\usepackage{xcolor}
\usepackage[most]{tcolorbox}
\usepackage[nameinlink,noabbrev]{cleveref}

\usepackage{bibunits} \defaultbibliographystyle{ACM-Reference-Format} \defaultbibliography{references}  

\lstdefinelanguage{Lean}{
  morekeywords={def,theorem,lemma,inductive,structure,class,where,namespace,end,by,forall,Prop,Type,match,with,if,then,else,let,in,do,axiom,constant,instance,import,abbrev},
  sensitive=true,
  morecomment=[l]{--},
  morecomment=[s]{/-}{-/},
  morestring=[b]",
  backgroundcolor=\color{gray!06}
}

 \definecolor{TheoremBackground}{gray}{0.95} 

\tcolorboxenvironment{theorem}{ enhanced,
  breakable,
  colback=TheoremBackground,
  colframe=TheoremBackground,
  boxrule=0pt,
  arc=1mm,
  left=0.9em,
  right=0.9em,
  top=0.65em,
  bottom=0.65em,
  before skip=0.9\baselineskip,
  after skip=0.9\baselineskip
}
\theoremstyle{theorem}

\newcommand{\NI}{\noindent}
\providecommand{\GIRY}{\ensuremath{\mathcal{G}}}

\newcommand{\SEMB}[1]{\lbrack\!\lbrack #1 \rbrack\!\rbrack}

\newenvironment{GRAMMAR}{\[\begin{array}{lcl}}{\end{array}\]}
\newenvironment{RULES}{\[\begin{array}{c}}{\end{array}\]}

\newcommand{\ACCESS}{\mathsf{Acc}}
\newcommand{\POWERSET}{\mathcal{P}}
\newcommand{\POWERSETFIN}{\POWERSET_{\text{\emph{fin}}}}

\newcommand{\EMPH}[1]{\emph{#1}}
\newcommand{\infer}[2]{\frac{\displaystyle{ #1 }}{\displaystyle{ #2 }}}
\newcommand{\ZEROPREMISERULE}[1]{\infer{-}{#1}}
\newcommand{\ONEPREMISERULE}[2]{\infer{#1}{#2}}

\newcommand{\RQONE}{{RQ1}}
\newcommand{\RQTWO}{{RQ2}}
\newcommand{\RQTHREE}{{RQ3}}
\newcommand{\RQFOUR}{{RQ4}}
\newcommand{\RQFIVE}{{RQ5}}

\newenvironment{FIGURE}{\begin{figure}\rule{\linewidth}{.5pt}}{\rule{\linewidth}{.5pt}\end{figure}}
\newcommand{\IMAGESIZE}[2]{\begin{center}\includegraphics[width=#1\textwidth]{images/#2}\end{center}}
\newcommand{\IMAGE}[1]{\IMAGESIZE{0.8}{#1}}

    \newcommand{\IE}{\textit{i.e.},\ } \newcommand{\EG}{e.g.,\ }    \newcommand{\STARTCF}{Cf.\ }

\newcommand{\lread}{\mathsf{read}}
\newcommand{\lwrite}{\mathsf{write}}

\providecommand{\AStepL}[1]{\mathrel{\xrightarrow{\,#1\,}_{\!a}}}
\providecommand{\BStepL}[1]{\mathrel{\xrightarrow{\,#1\,}_{\!b}}}
\providecommand{\PStepL}[1]{\mathrel{\xrightarrow{\,#1\,}_{\!p}}}
\providecommand{\AStepT}[1]{\mathrel{\xRightarrow{\,#1\,}_{\!a}}}
\providecommand{\BStepT}[1]{\mathrel{\xRightarrow{\,#1\,}_{\!b}}}
\providecommand{\PStepT}[1]{\mathrel{\xRightarrow{#1}_{\!p}}}
\providecommand{\lread}{\mathsf{read}}
\providecommand{\lwrite}{\mathsf{write}}

\providecommand{\bind}{\mathbin{\gg\!\!=}}

\newcommand{\RHAStep}{\mathrel{\rightsquigarrow_{\!a}^{\mathsf{RH}}}}
\newcommand{\RHBStep}{\mathrel{\rightsquigarrow_{\!b}^{\mathsf{RH}}}}
\newcommand{\RHPStep}{\mathrel{\rightsquigarrow_{\!p}^{\mathsf{RH}}}}
\newcommand{\RHStep}{\ensuremath{\mathsf{step}_{br,locs}^{\mathsf{RH}}}}
\newcommand{\RHStepHat}{\ensuremath{\widehat{\mathsf{step}}_{br,locs}^{\mathsf{RH}}}}
\newcommand{\RunRH}[1]{\ensuremath{\mathsf{run}^{#1}_{br,locs}}}
\newcommand{\RunRHT}[1]{\ensuremath{\mathsf{run}^{#1}_{br,locs,\mathsf{tr}}}}
\newcommand{\RunDT}[1]{\ensuremath{\mathsf{run}^{#1}_{\mathsf{det},\mathsf{tr}}}}
\newcommand{\protV}{\ensuremath{\pi_{locs}}}

\newcommand{\dirac}{\ensuremath{\delta}}
\newcommand{\FDist}{\ensuremath{\mathsf{Dist}}}
\newcommand{\PMF}{\ensuremath{\mathsf{PMF}}}
\newcommand{\Mem}{\ensuremath{\mathsf{Memory}}}
\newcommand{\Loc}{\ensuremath{\mathsf{Loc}}}
\newcommand{\INT}{\mathbb{Z}}
\newcommand{\Val}{\ensuremath{\mathsf{Val}}}
\newcommand{\adj}{\ensuremath{\mathsf{adj}}}
\newcommand{\safe}{\ensuremath{\mathsf{safe}}}
\newcommand{\wf}{\ensuremath{\mathsf{wf}}}
\newcommand{\pflip}{\ensuremath{\mathsf{pflip}}}
\newcommand{\EvalOrd}[1]{\ensuremath{\llbracket #1 \rrbracket^{\mathsf{ord}}}}
\newcommand{\Nat}{\ensuremath{\mathbb{N}}}

\newcommand{\keep}{\ensuremath{\mathsf{keep}}}

\newcommand{\Low}{\ensuremath{\mathsf{L}}}
\newcommand{\High}{\ensuremath{\mathsf{H}}}

\newcommand{\SecLevel}{\ensuremath{\mathsf{Sec}}}
\newcommand{\LowView}{\ensuremath{\mathsf{lowView}}}
\newcommand{\HighView}{\ensuremath{\mathsf{highView}}}
\newcommand{\PView}{\ensuremath{\mathsf{PView}}}
\newcommand{\ProtView}{\ensuremath{\mathsf{protView}}}
\newcommand{\AccessTrace}{\ensuremath{\mathsf{accessTrace}}}
\newcommand{\RunDTrace}[1]{\ensuremath{\mathsf{runDTrace}^{#1}}}
\newcommand{\RunRHTrace}[1]{\ensuremath{\mathsf{runRHTrace}^{#1}}}
\newcommand{\ObsOrd}[1]{\ensuremath{\mathsf{ObsOrd}^{#1}}}
\newcommand{\ObsRH}[1]{\ensuremath{\mathsf{ObsRH}^{#1}}}

\newcommand{\Admissible}{\ensuremath{\mathsf{Admissible}}}
\newcommand{\LowAccessTrace}{\ensuremath{\mathsf{lowAccessTrace}}}
\newcommand{\Terminal}{\ensuremath{\mathsf{terminal}}}

\newcommand{\Agree}{\ensuremath{\mathsf{Agree}}}
\newcommand{\DProtectedView}[1]{\ensuremath{\mathsf{DView}^{#1}}}
\newcommand{\LowEqRel}[1]{\mathrel{\approx_{#1}}}
\newcommand{\NoFlip}{\ensuremath{\mathsf{noFlip}}}
\newcommand{\OrdNI}{\ensuremath{\mathsf{OrdinaryNI}}}
\newcommand{\RHNI}{\ensuremath{\mathsf{RowhammerNI}}}
\newcommand{\RHProtectedView}[1]{\ensuremath{\mathsf{RHView}^{#1}}}
\newcommand{\RobustRHNI}{\ensuremath{\mathsf{RobustRowhammerNI}}}
 
\begin{document}

\title[Mechanised operational semantics of Rowhammer]{Mechanised operational semantics of Rowhammer}

\author{Martin Berger}
\email{contact@martinfriedrichberger.net}
\orcid{https://orcid.org/0000-0003-3239-5812}
\affiliation{
  \institution{University of Sussex}
  \city{Brighton}
  \country{UK}
}
\affiliation{
  \institution{Montanarius Ltd}
  \city{London}
  \country{UK}
}

\author{Amir Naseredini}
\email{sahnaseredini@gmail.com}
\affiliation{
  \institution{Independent Researcher}
  \city{Cambridge}
  \country{UK}
}

\renewcommand{\shortauthors}{M.~Berger \& A.~Naseredini}  \begin{abstract}

Rowhammer is a hardware vulnerability in dynamic random-access memory
(DRAM) in which repeated accesses to one or more aggressor rows can
induce bit-flips in nearby victim rows. This phenomenon violates a
core assumption of conventional programming language semantics: that
reading from or writing to one memory location does not modify others.
Despite the security importance of this phenomenon, there
is no widely established formal framework connecting Rowhammer faults
with program behaviour.  This makes it difficult to reason rigorously
about the efficacy of proposed Rowhammer defences and the
program-level guarantees they provide. To address this gap, we present a probabilistic
small-step operational semantics for an idealised imperative language
subject to Rowhammer-style faults.  The semantics is
high-level in that it abstracts from DRAM internals and semiconductor
physics. A general probabilistic fault model parameterises the
semantics, representing Rowhammer-style faults by assigning probabilities to
bit-flips during read or write operations within a specified victim region.
The resulting distributions are propagated through programs using the
standard monadic structure of probabilistic computation. As a case study,
we formalise physical separation, a well-known defence that
places program variables sufficiently far apart in physical memory
that an access to one variable cannot disturb another.  We prove a
distribution-independent semantic collapse theorem: for every finite
execution, including prefixes of terminating and non-terminating
executions, the protected projection of the probabilistic Rowhammer
semantics is the Dirac distribution of the corresponding Rowhammer-free
execution.  Furthermore, we develop an observation-parametric
account of secure information flow.  Non-interference is expressed as
a hyperproperty comparing the distributions of low observations from
low-equivalent initial memories.  Under physical separation, ordinary
non-interference and probability-sensitive Rowhammer non-interference
coincide for every observer of protected behaviour.  Consequently,
physical separation preserves non-interference for every admissible
fault model, while every Rowhammer non-interference violation
reflects a violation already present in the  Rowhammer-free semantics. The
development  is fully mechanised in Lean using mathlib, relying on
no unfinished proofs or problem-specific axioms.

\end{abstract}

\begin{CCSXML}
<ccs2012>
   <concept>
       <concept_id>10003752.10010124.10010131</concept_id>
       <concept_desc>Theory of computation~Program semantics</concept_desc>
       <concept_significance>500</concept_significance>
       </concept>
   <concept>
       <concept_id>10011007.10010940.10010992.10010998</concept_id>
       <concept_desc>Software and its engineering~Formal methods</concept_desc>
       <concept_significance>500</concept_significance>
       </concept>
   <concept>
       <concept_id>10002978.10003006.10011608</concept_id>
       <concept_desc>Security and privacy~Information flow control</concept_desc>
       <concept_significance>300</concept_significance>
       </concept>
   <concept>
       <concept_id>10002978.10003006</concept_id>
       <concept_desc>Security and privacy~Systems security</concept_desc>
       <concept_significance>300</concept_significance>
       </concept>
 </ccs2012>
\end{CCSXML}

\ccsdesc[500]{Theory of computation~Program semantics}
\ccsdesc[500]{Software and its engineering~Formal methods}
\ccsdesc[300]{Security and privacy~Information flow control}
\ccsdesc[300]{Security and privacy~Systems security}
\ccsdesc[500]{Security and privacy~Logic and verification}

\keywords{Probability theory, Probabilistic fault models, Rowhammer,
  Programming language semantics, Operational semantics, Relational
  safety, Mechanised verification, Lean 4, Secure information flow,
  Non-interference.}
 
\maketitle

\section{Introduction}
\label{section_introduction}

Most reasoning about software today relies on a locality assumption:
\begin{quote}
  \EMPH{Reading or writing one memory location does not modify other
    locations.}
\end{quote}
Rowhammer violates this assumption.  Repeated activation of
one or more aggressor rows in dynamic random-access memory (DRAM) can
induce bit-flips in physically nearby victim rows
\cite{kim2014rowhammer,seaborn2015exploiting}.  The resulting mismatch
between the memory model assumed by software semantics and the
behaviour of its semiconductor implementation creates a semantic gap:
hardware research describes disturbance mechanisms, address mappings,
and mitigations, whereas software verification reasons about commands,
stores, traces, and observations.
This gap is particularly important for security.  Rowhammer is usually
presented as an integrity failure, but corruption of a high-security
value can alter control flow or a later low-security output, while
corruption of a low-security location can directly change an
attacker-visible result.  Existing theories of secure information flow
and non-interference
\cite{goguen1982security,sabelfeld2003language} therefore cannot simply
be applied unchanged: they require a semantics that makes the
probabilistic, non-local effect of memory accesses explicit.
This paper asks closely connected research questions.
\begin{itemize}

\item \textbf{\RQONE.} Can Rowhammer be modelled compositionally in
  the familiar language of operational semantics, without modelling
  device physics?

\item \textbf{\RQTWO.} Can the Rowhammer fault mechanism be captured
  by a small, conceptually non-intrusive semantic interface?

\item \textbf{\RQTHREE.} Can such an interface be made substantially
  independent of the target programming language?

\item \textbf{\RQFOUR.} Can such a model express and justify a concrete
  Rowhammer defence?

\item \textbf{\RQFIVE.} Can the abstraction make Rowhammer amenable
  to existing semantic security frameworks, such as information flow
  and non-interference reasoning?

\end{itemize}
We answer all questions affirmatively.

\subsection{The semantic view: abstracting memory access}
A standard way to give semantics to an imperative language is to treat
programs as state transformers.  In its most concrete form, assignment
and variable lookup are explained by clauses such as
\[
  \SEMB{x := e}(m)
  =
  m[x\mapsto v]
  \qquad\text{where}\qquad
  \SEMB{e}(m)=v,
\]
and
\[
  \SEMB{x}(m)=m(x).
  \]
where $m$ ranges over memories and $e$ over expressions.  These equations present memory as a
concrete map and use two primitive operations on that map: read the
value stored at a location, and update the value stored at a location.
A recurring lesson of programming language semantics is that these two
operations should be made explicit.  Once lookup and update are
abstracted from the concrete function space $\Mem = \Loc\to\Val$,
where $\Loc$ are locations in memory, and $\Val$ the
values being stored, the rest of the language semantics no longer has to know
how memory is represented.  The same assignment rule can then be
interpreted over different models of store, aliasing, heaps,
concurrency, failure, or other effects.

This idea has appeared
repeatedly, under different names.  McCarthy's select/store axioms for
arrays isolate reading and updating as primitive operations on
memories~\cite{mccarthy1962mathematical}.  Monadic and
algebraic-effect semantics make the separation systematic:
computations with state are interpreted as effectful maps, and state
itself may be presented by operations such as lookup and update
together with their laws~\cite{moggi1991notions,plotkin2002notions}.
Modular operational semantics pursues a related goal from the
small-step side, factoring stores and other auxiliary components out
of individual transition rules so that language constructs can be
specified more independently~\cite{mosses2004modular}.  The common
pattern is that memory access becomes an interface rather than an
implementation detail.

Our Rowhammer semantics uses this idea.  Ordinary memory operations
are deterministic:
\[
   \lread : (\Loc \times \Mem) \to \Val
      \qquad
   \lwrite : (\Loc \times \Val \times \Mem) \to \Mem
\]
Rowhammer-affected read and write are instead probabilistic
operations\footnote{We discuss in
\Cref{section_probabilistic_model_of_rowhammer} why the codomain of
$\lread$ needs to be $\FDist(\Val\times\Mem)$ and not
$\Val\times\FDist(\Mem)$.}:
\[
   \lread : (\Loc \times \Mem)\to \FDist(\Val\times\Mem)
      \quad
   \lwrite : (\Loc \times \Val \times \Mem) \to \FDist(\Mem)
\]
The operations still have the shape of a local memory access, but
their results are distributions that include the post-access memory.
The distributions model Rowhammer faults.

The rest of the semantics is then obtained monadically:
deterministic steps are embedded as Dirac distributions, and compound
computations are sequenced by monadic bind for the probability monad.
Thus the Rowhammer model is not a separate hardware semantics bolted
onto the language.  It is a small probabilistic reinterpretation of the
read/write interface already present in ordinary state-transformer
semantics.  This is the key abstraction boundary used throughout the
paper: Rowhammer enters only at memory accesses, the surrounding
operational semantics lifts compositionally around that local
probabilistic effect.

This move from deterministic to probabilistic state transformers is not unique to
Rowhammer: deterministic Turing machines and programming languages naturally become
probabilistic by allowing programs to branch according to a probability distribution.
In those settings, probabilistic choice is an intentional feature of the
syntax or model. Rowhammer is different. Its probabilistic
behaviour is an unwanted, observable consequence of the physical substrate, making it
conceptually closer to Shannon's noisy-channel model from information theory: an
intended memory operation is transmitted through an unreliable medium and may be
received as corrupted state. In contrast, this noise is not
merely ambient: it is actively shaped by an attacker's chosen access patterns and
strictly constrained by physical locality in DRAM. Our semantic task is therefore to
expose this hardware-level unreliability without modifying the language syntax,
placing a probabilistic fault model directly behind the ordinary read/write interface.
Ultimately, the novelty lies not in using probabilistic state transformers per se, but
in the observation that Rowhammer can be elegantly isolated as a probabilistic
reinterpretation of memory access.

\subsection{Contributions}
To address the research questions, we make the following contributions:
\begin{itemize}

\item A fine-grained labelled operational semantics in which Rowhammer
  has a small, explicit interface at reads and writes.

\item A general probabilistic fault model separating spatial
  confinement from numerical assumptions about bit-flips.

\item A distribution-independent soundness theorem for physical
  separation, a well-known Rowhammer defence, covering terminating and
  non-terminating executions.

\item An observation-parametric theorem that physical separation
  preserves and reflects probability-sensitive non-interference.

\end{itemize}
The complete development is mechanised in Lean \cite{moura2021lean4}
using mathlib \cite{mathlib2020}. The development contains no
unfinished proofs or problem-specific axioms.  An axiom audit
(\texttt{scripts/axiomAudit.sh} in the repo) confirms that all
results rest only on Lean's three foundational axioms
(\texttt{propext}, \texttt{Classical.choice}, \texttt{Quot.sound}): no
probability axioms are added.  The Lean excerpts shown are
automatically extracted verbatim from the sources so they cannot drift
from the development.  All Lean code is open-sourced and anonymised at
\cite{NaserediniA:leanProofsSeparationANONYMISED}.  Full operational
rules, auxiliary metatheory, and extended mechanisation details are
provided in the supplementary material. The appendices cited in the
text are provided as anonymised supplementary material.
 \section{Introduction to Rowhammer}\label{section_rowhammer}

To appreciate what the present work does and does not promise, it
helps to understand the physical mechanism it abstracts.  The rest of
the paper relies on only a few concepts from DRAM architecture, rows,
activation, refresh, and the locality of disturbance, and we introduce
them here from first principles.  Readers already familiar with
Rowhammer may skip to the examples at the end of this section.  For
the underlying semiconductor physics we refer to
\cite{WalkerA:ondramratpoi}. For a broad survey of the attack
landscape see \cite{mutlu2019retrospective}.

\paragraph{The memory hierarchy.}
DRAM occupies an important point in a trade-off
(\Cref{fig:memory-hierarchy}).  Above it sit registers and caches:
fast, but expensive.  Below sits persistent storage: vast, but orders
of magnitude slower.  DRAM's role is to be as large and as cheap as
possible while remaining just fast enough. It achieves this with
only one transistor and one capacitor per bit, against the six
transistors of a cache cell.  This economic pressure towards density
is Rowhammer's root cause.  The more tightly cells are packed, the
stronger the electrical coupling between them, which is why the
problem has worsened with every DRAM generation
\cite{kim2014rowhammer}.

\begin{figure}
  \centering
  \begin{tikzpicture}[
      lvl/.style={draw=black!60, minimum height=0.62cm, font=\small,
                  inner sep=2pt, align=center},
      side/.style={font=\small\itshape, align=center}]
    \node[lvl, minimum width=2.2cm, fill=black!8]
      (l1) at (0, 2.1) {registers};
    \node[lvl, minimum width=4.0cm, fill=black!8]
      (l2) at (0, 1.4) {caches (L1--L3)};
    \node[lvl, minimum width=5.8cm, fill=blue!20, thick]
      (l3) at (0, 0.7) {DRAM main memory};
    \node[lvl, minimum width=7.6cm, fill=black!8]
      (l4) at (0, 0)   {persistent storage (SSD, disk)};
    \draw[->, black!70] (-4.4, -0.2) -- (-4.4, 2.3)
      node[side, midway, left, xshift=-2pt] {faster,\\costlier\\per bit};
    \draw[->, black!70] (4.4, 2.3) -- (4.4, -0.2)
      node[side, midway, right, xshift=2pt] {larger,\\cheaper\\per bit};
  \end{tikzpicture}
  \caption{The memory hierarchy.  Each level trades speed for
    capacity.  DRAM (highlighted) wins its place through cell
    minimalism, one transistor and one capacitor per bit, and the
    resulting drive towards ever-denser packing is the root cause of
    the disturbance effects this paper models.}
  \label{fig:memory-hierarchy}
\end{figure}
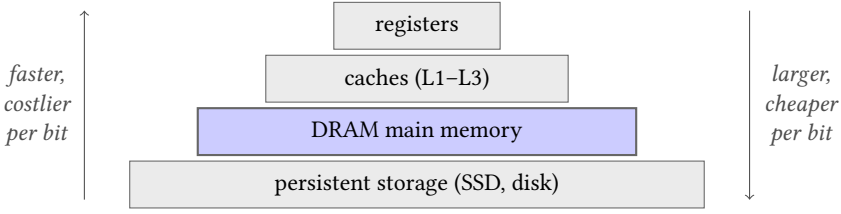

\paragraph{How DRAM stores data.}
At the hardware level, each DRAM bit is stored as electrical charge in
a capacitor, guarded by an access transistor.  Cells are organised
into large two-dimensional grids called \emph{banks}, which operate
independently of one another.  Within a bank, a horizontal line of
cells, a \emph{row}, typically a few kilobytes wide, shares a common
control wire known as the \emph{wordline}, while cells in the same
column share a \emph{bitline}.  Individual cells cannot be addressed
directly.  To access any cell, the memory controller must
\emph{activate} the row containing it: raising the row's wordline
voltage connects every cell of the row to its bitline and copies the
row's contents into the bank's \emph{row buffer}, where reads and
writes then take place.  Accessing a different row of the same bank
requires closing the currently open row, meaning lowering its wordline
voltage before activating the new one.
Capacitors leak charge.  Left alone, a DRAM cell loses its stored
charge within a fraction of a second, so the memory controller
periodically rewrites every row---a \emph{refresh}.  In commodity
DRAM, every row is typically guaranteed a refresh once every 64\,ms
\cite{kim2014rowhammer}.  The reliability of DRAM therefore rests on a
quantitative margin: between two refreshes, a cell's charge must stay
on the correct side of a threshold, despite whatever electrical
disturbance its neighbourhood produces.

\paragraph{The Rowhammer effect.}
Kim et al.~\cite{kim2014rowhammer} showed in 2014 that this margin can
be exhausted \emph{deliberately}.  Every activation of a row causes
slight voltage fluctuations on its wordline, and these fluctuations
accelerate charge leakage in the cells of \emph{physically adjacent}
rows, an electromagnetic coupling effect that grows stronger as cells
are packed more densely.  A single activation is harmless.  But a
program that activates the same row hundreds of thousands of times
within one refresh interval, \emph{hammering} it, can drain a
neighbouring cell past its threshold before the next refresh arrives.
The value read thereafter is not the value last written: a bit has
flipped in a row that was never written to, or even addressed, by the
program.  The repeatedly activated rows are called \emph{aggressor}
rows and the corrupted ones \emph{victim} rows
(\Cref{fig:blast-radius}).  The minimum number of activations needed
is called the \emph{Rowhammer threshold}: it has dropped sharply as
DRAM density has increased, from over a hundred thousand activations
on DDR3 to tens of thousands on modern chips
\cite{kim2014rowhammer,frigo2020trrespass}, which is why the problem
is getting worse, not better, with each DRAM generation.

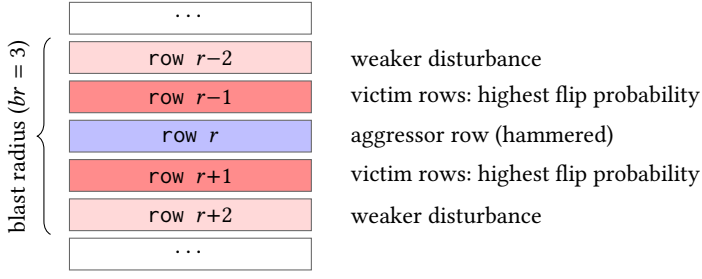
\begin{figure}
  \centering
  \begin{tikzpicture}[
      row/.style={draw=black!60, minimum width=3.2cm, minimum height=0.42cm,
                  font=\small\ttfamily, inner sep=1pt},
      note/.style={font=\small, anchor=west},
      yscale=0.52]
\node[row, fill=white]    (r0) at (0, 4) {$\cdots$};
    \node[row, fill=red!15]   (r1) at (0, 3) {row $r{-}2$};
    \node[row, fill=red!45]   (r2) at (0, 2) {row $r{-}1$};
    \node[row, fill=blue!25]  (r3) at (0, 1) {row $r$};
    \node[row, fill=red!45]   (r4) at (0, 0) {row $r{+}1$};
    \node[row, fill=red!15]   (r5) at (0,-1) {row $r{+}2$};
    \node[row, fill=white]    (r6) at (0,-2) {$\cdots$};
\node[note] at (2.0, 1)  {aggressor row (hammered)};
    \node[note] at (2.0, 2)  {victim rows: highest flip probability};
    \node[note] at (2.0, 0)  {victim rows: highest flip probability};
    \node[note] at (2.0, 3)  {weaker disturbance};
    \node[note] at (2.0,-1)  {weaker disturbance};
\draw[decorate, decoration={brace,  mirror, amplitude=5pt}]
      (-1.85, 3.5) -- (-1.85, -1.5)
      node[midway, xshift=-12pt, rotate=90, font=\small]
      {blast radius ($br = 3$)};
  \end{tikzpicture}
  \caption{Disturbance is local.  Repeatedly activating
    (\emph{hammering}) row $r$ accelerates charge leakage in nearby
    rows. The flip probability is highest for the immediate
    neighbours $r\pm1$ and decays with distance.  Beyond the blast
    radius $br$ the effect is negligible.  The victim rows are
    corrupted although the program never writes to or even
    addresses them.}
  \label{fig:blast-radius}
\end{figure}

Rowhammer is not an exotic laboratory effect.  It is exploitable from
ordinary unprivileged code, including JavaScript in a web browser
\cite{gruss2016rowhammerjs}, and flipping a single page-table or
capability bit can escalate one hardware fault into arbitrary memory
access \cite{seaborn2015exploiting}.  The effect has also grown richer
since its discovery.  Half-Double \cite{kogler2022halfdouble} induces
flips beyond immediate neighbours by combining accesses at different
distances, TRRespass \cite{frigo2020trrespass} and Blacksmith
\cite{jattke2022blacksmith} bypass in-DRAM mitigations with many-sided
and non-uniform hammering patterns, and RAMBleed
\cite{kwong2020rambleed} turns the data-dependence of flip
probabilities into a \emph{read} primitive, making Rowhammer a
confidentiality problem as well as an integrity problem.

\paragraph{Locality and the blast radius.}
The one saving grace of DRAM physics is its strong spatial locality:
the disturbance an activation causes decays rapidly with physical
distance, and beyond a few rows it is negligible.  In the rest of the
text we call the maximum distance at which an access can disturb a
cell the \emph{blast radius}, abbreviated $br$.  The blast radius of a
given DRAM device can be approximated empirically, and it is the key
intuition behind an entire class of defences: if everything worth
protecting is kept at least $br$ rows away from everything an attacker
can hammer, the disturbance never reaches it.

\paragraph{Physical separation defences.}
Defences based on \emph{placement} exploit exactly this locality.
CATT \cite{brasser2017catt} partitions kernel from user memory with
unused \emph{guard rows}, so that user-space hammering cannot reach
kernel rows. ZebRAM \cite{konoth2018zebram} interleaves guard rows
between all data rows (\Cref{fig:separation}).  The correctness claim
of such defences is physical: every row an attacker can reach is
separated from every row worth protecting by more than the
disturbance radius.  In this paper the discipline is captured by a
predicate $\safe(br,locs)$: the victim region generated by an access
to any protected location is disjoint from the protected set $locs$.
In the distance-based instance this means that distinct protected
rows lie at least $br$ apart. But separation is required \emph{within}
the protected set too, because the program's own accesses hammer
their own neighbourhoods.  The guard-row layouts above are the same
discipline applied to the whole address space, with the
attacker-reachable region kept at distance $br$ from everything
protected.

\begin{figure}
  \centering
  \begin{tikzpicture}[
      row/.style={draw=black!60, minimum width=3.2cm, minimum height=0.42cm,
                  font=\small, inner sep=1pt},
      note/.style={font=\small, anchor=west},
      yscale=0.52]
    \node[row, fill=blue!25]  (a0) at (0, 5) {attacker-reachable};
    \node[row, fill=blue!25]  (a1) at (0, 4) {attacker-reachable};
    \node[row, fill=black!10] (g1) at (0, 3) {guard row};
    \node[row, fill=black!10] (g2) at (0, 2) {guard row};
    \node[row, fill=green!25] (p1) at (0, 1) {protected ($x$)};
    \node[row, fill=black!10] (g3) at (0, 0) {guard row};
    \node[row, fill=black!10] (g4) at (0,-1) {guard row};
    \node[row, fill=green!25] (p2) at (0,-2) {protected ($y$)};
    \draw[decorate, decoration={brace, amplitude=5pt}]
      (1.85, 3.5) -- (1.85, 1.5)
      node[midway, xshift=16pt, font=\small] {$\geq br$};
    \draw[decorate, decoration={brace, amplitude=5pt}]
      (1.85, 0.5) -- (1.85, -1.5)
      node[midway, xshift=16pt, font=\small] {$\geq br$};
    \node[note] at (-6.8, 4.5) {hammering here \dots};
    \draw[->, black!70] (-3.9, 4.5) -- (-1.9, 4.5);
    \node[note] at (-7.0, 1) {\dots cannot reach here};
    \draw[->, black!70] (-3.9, 1) -- (-1.9, 1);
  \end{tikzpicture}
  \caption{Physical separation with a blast radius of $br=3$.  Guard
    rows keep every protected row at least $br$ rows away from every
    row the program (or an attacker) accesses, including the
    \emph{other} protected rows, since accesses to $x$ hammer $x$'s own
    neighbourhood.  This is the layout discipline that the predicate
    $\safe(br,locs)$ of \Cref{sec:fault-kernel-contract} abstracts.}
  \label{fig:separation}
\end{figure}
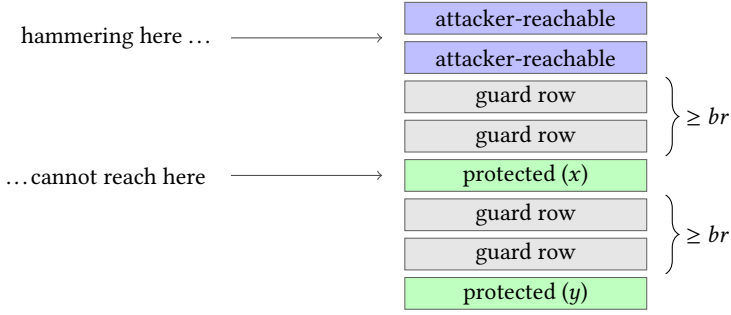

\paragraph{The semantic gap.}
Program verification speaks a different language: commands, stores,
traces, observables.  A proof about a program is a proof about its
semantics, and standard semantics assume that memory changes only
when written.  Rowhammer breaks exactly that assumption, and
placement defences restore it, but no theorem connects the two
levels.  What is missing is a semantics in which the physical claim
(``flips are confined to a blast radius'') and the placement claim
(``protected rows are outside every blast radius'') can both be
stated, and a theorem that together they restore the ordinary meaning
of programs.  This paper supplies both.  The examples below, adapted
from  \cite{naseredini2024towards}, illustrate what goes
wrong without separation, and, read against \Cref{thm:goal-formal},
identify exactly which hypothesis each failure violates.

\begin{example}[Corruption without a write]
\label{ex:corruption-without-write}
Let $x$ and $y$ occupy adjacent rows, and let the program hammer $y$
while never assigning to $x$:
\[
  x := 1;\
  \mathsf{while}\ (0=0)\ \mathsf{do}\ y := y + 1 .
\]
Every access to $y$ disturbs the blast radius of $y$'s row, which
contains $x$.  In every ordinary semantics $x$ holds $1$ forever. Under
Rowhammer, the probability that $x$ still holds $1$ decays with every
iteration.  No analysis that equates ``unwritten'' with ``unchanged''
can see this.
\end{example}

\begin{example}[A dead branch awakens]
\label{ex:dead-branch}
Corruption is not limited to data, it can also redirect control.  Let $x$ and
$y$ be adjacent and consider
\[
  x := y;\
  x := x;\
  \mathsf{if}\ x = y\ \mathsf{then}\ \mathsf{skip}\ \mathsf{else}\ P .
\]
Ordinarily the guard is always true and $P$ is dead code.  But the
reads and writes of $x$ hammer $x$'s neighbourhood, which contains
$y$: a flip in $y$ between the first assignment and the guard makes
the guard false, and the machine executes $P$.  Which code runs is now
probabilistic.  This is why the protected view of
\Cref{sec:protected-view} tracks the \emph{residual program} and the
\emph{access trace}, not just memory: a defence that preserved memory
but not control flow would not restore the meaning of programs.  Note
also that both $x$ and $y$ here are the program's \emph{own}
variables: if the protected set $locs=\{x,y\}$ places them in
adjacent rows, $\safe(br,locs)$ fails, and the theorem correctly does
not apply: separation is required \emph{within} the protected set,
not only between the program and the outside world.
\end{example}

\begin{example}[An unowned access]
\label{ex:unowned-access}
Separation of the program's own rows is not enough either.  Suppose
$locs=\{x\}$ is safely laid out, but the program also touches a row
$z\notin locs$ that happens to be adjacent to $x$:
\[
  z := z + 1;\
  \mathsf{if}\ x = 0\ \mathsf{then}\ \mathsf{skip}\ \mathsf{else}\ P .
\]
The access to $z$ triggers a fault in $z$'s blast radius, about which
$\safe(br,locs)$ says nothing, and that radius contains $x$.  This is
what the well-formedness hypothesis $\wf_{locs}(P)$ of
\Cref{thm:goal-formal} rules out: every location the program touches
must be part of the protected layout, so that every fault the program
can trigger is one the separation condition constrains.
\end{example}

These examples give us the shape of the theorem we need: quantified over
fault behaviour (flip probabilities are the attacker's, not ours),
conclusion covering memory, control, and accesses alike, and
hypotheses (frame-locality, separation, well-formedness) each of
which is violated by some concrete program above when dropped.
 \section{Mathematical preliminaries}\label{section_mathematical_preliminaries}

We write $\POWERSET(S)$ for the \EMPH{powerset} of a set $S$ and
$\POWERSETFIN(S)$ for the set of all finite subsets of $S$.
The Rowhammer semantics assigns to each transition a
\emph{distribution} over successor configurations.  This section fixes
the small amount of probability theory involved. Readers familiar with
the probability monad (also known as Giry monad)
\cite{giry1982categorical,LawvereFW:catopromap} or with monadic
treatments of probabilistic programs
\cite{kozen1981semantics,ramsey2002stochastic,DahlqvistF:sempropagi}
can skim it for notation. All
probability spaces in this paper will be discrete, \IE finite or
countably infinite, so all occurring measure spaces are trivial and we
shall not mention them.  For a type $\alpha$, we write
$\FDist(\alpha)$ for the type of discrete probability distributions on
$\alpha$, \IE functions $\mu:\alpha\to[0,1]$ with $\sum_{x}\mu(x) =
1$. The \emph{support} of $\mu$, written $\mathrm{supp}(\mu)$, is the
set of $x$ with $\mu(x)>0$. The \EMPH{Dirac distribution}, also called
\EMPH{point mass}, at $x\in\alpha$ is the distribution
$\delta_x\in\FDist(\alpha)$ given by $\delta_x(y) = 1$ if $y=x$, and
$\delta_x(y)=0$ otherwise, for all $y \in \alpha$.
A \emph{probability kernel} from $\alpha$ to $\beta$ is an
input-indexed family of distributions, \IE a function $\kappa
:\alpha\to\FDist(\beta).$ For each input $a\in\alpha$, the value
$\kappa (a)$ is therefore a distribution on $\beta$. Any distribution
$\mu\in\FDist(\beta)$ determines the constant kernel $a\mapsto\mu$.
Likewise, a deterministic function $f:\alpha\to\beta$ embeds into the
probabilistic setting as the Dirac kernel $a\mapsto\dirac(f(a))$.  In
this paper, ``kernel'' always means a discrete probability kernel of
the form $\alpha\to\FDist(\beta)$. In the Lean development this is
represented as a function returning a $\PMF$ (probability mass function), rather than by mathlib's
more general measure-theoretic type \texttt{ProbabilityTheory.Kernel}.
Sequential composition of probabilistic computations is monadic.  For
our purposes, a monad consists of a type constructor $M$ together with
operations \EMPH{unit} $\mathsf{return}_{\alpha}:\alpha\to M(\alpha)$
and \EMPH{bind} $\bind: M(\alpha)\to\bigl(\alpha\to M(\beta)\bigr)\to
M(\beta)$, satisfying the left-identity, right-identity, and
associativity laws
\cite{maclane1998categories,moggi1991notions}. Intuitively,
$M(\alpha)$ is a type of computations producing values of type
$\alpha$, $\mathsf{return}$ embeds a pure value, and bind sequences.
The classical probability monad $\GIRY$ instantiates this structure
with probability measures on measurable spaces.  Since all spaces in
this paper are discrete, we work with its discrete analogue, spelled
out next.
On objects it is given by $\GIRY(\alpha)\triangleq\FDist(\alpha).$ On
morphisms, \GIRY\ acts by pushforward of distributions: for $f :
\alpha \rightarrow \beta$, we have $\GIRY(f) : \FDist(\alpha) \rightarrow
\FDist(\beta)$, and for  $\mu \in \FDist(\alpha)$, we set $ \GIRY(f)(\mu)(y)
= \sum_{\substack{x\in\alpha, f(x)=y}} \mu(x)$.  The unit
$\mathsf{return}_{\alpha}:\alpha \rightarrow \GIRY(\alpha)$ is the
Dirac distribution $\delta_x$, in other words, the point mass
concentrated at $x$.  $\GIRY$'s bind operation, given a distribution $
\mu\in\FDist(\alpha) $ and a probability kernel $
\kappa:\alpha\to\FDist(\beta), $ is written
\[
   \mathbf{let}\ x\leftarrow\mu\ \mathbf{in}\ \kappa(x),
\]
and defined pointwise by $ \left(
\mathbf{let}\ x\leftarrow\mu\ \mathbf{in}\ \kappa(x) \right)(y) =
\sum_{x\in\alpha}\mu(x)\,\kappa(x)(y).$ In the Lean development,
$\GIRY$ is represented by $\PMF$.
The Lean development mirrors this notation in the file
\texttt{DiscreteDist.lean}.  There, \texttt{Dist} is an abbreviation
for mathlib's \texttt{PMF}, the type of discrete probability mass
functions.  Thus the formalisation does not assume finite support. The
file provides the small interface used in the proofs: Dirac
distributions, bind, map, support lemmas, and injectivity of Dirac.
 \section{A deterministic small-step operational semantics for a WHILE language}\label{section_language}

The aim of this paper is to express Rowhammer in terms of programming
language semantics, so as to be able to reason about program
correctness and the efficacy of Rowhammer defences using the extensive
tool-sets developed over the previous decades by the programming
language and verification communities, but without having to go
through semiconductor physics, or even knowledge of DRAM
organisation. For this purpose we need a paradigmatic programming
language that is simple enough to allow us to explain the key ideas
without being drowned out by language complexity, but at the same time
general enough to cover interesting Rowhammer-related phenomena. We
choose the well-known WHILE language
\cite{vasconcelos2016language,RiisNielsonH:semwitaaa}. Its syntax is
given by the following rules.
\begin{GRAMMAR}
  e &::=& c \mid x \mid e+e
  \\[2mm]
  b &::=&
    \mathsf{true}
    \mid \mathsf{false}
    \mid \neg b
    \mid b\wedge b
    \mid e=e
  \\[2mm]
  P &::=&
    \mathsf{skip}
    \mid x:=e
    \mid P; P
    \mid \mathsf{if}\ b\ \mathsf{then}\ P\ \mathsf{else}\ P
    \mid \mathsf{while}\ b\ \mathsf{do}\ P
\end{GRAMMAR}

\NI Here $c$ ranges over integers, and $x$ over memory locations,
collectively denoted $\Loc$. It is natural to assume that $\Loc =
\INT$, but nothing in our development depends on that. Since the
language lacks higher-order functions and dynamic memory allocation,
we can associate with each $e$, $b$ and $P$ the set of locations
syntactically mentioned.  This gives rise to the \texttt{WellFormed}
predicates in our Lean formalisation: \EG \texttt{p.WellFormed locs}
means that the locations occurring syntactically in the program
\texttt{p} are contained in the set \texttt{locs}. Clearly
well-formedness is preserved by reduction.

We give a conventional small-step operational semantics with two
minor twists.  The operational semantics is a deliberately
\emph{fine-grained} small-step semantics: arithmetic expressions,
Boolean expressions, and programs all reduce one step at a time, so
every memory read and every write is an individual transition.
Transitions carry a \emph{label} recording the access they perform:
\[
  \ell \ ::=\  \lread(x) \mid \lwrite(x) \mid \tau .
\]
A read of variable $x$ is labelled $\lread(x)$, the write performed by
an assignment to $x$ is labelled $\lwrite(x)$, and every other
transition, operator reduction, branch selection, sequencing, loop
unfolding, is silent ($\tau$). We write $\ACCESS$ for the set of all such labels.
The labels play no role in the
deterministic semantics itself (only assignment writes modify
memory) but help injecting  the Rowhammer faults
(\Cref{section_probabilistic_model_of_rowhammer}).
Congruence rules propagate the label of their
premise.  We use three labelled judgements on configurations
\[
  \langle e,m\rangle \AStepL{\ell} \langle e',m'\rangle
  \qquad\qquad
  \langle b,m\rangle \BStepL{\ell} \langle b',m'\rangle
  \qquad\qquad
  \langle P,m\rangle \PStepL{\ell} \langle P',m'\rangle
\]
In all cases, $m$ ranges over memories, \IE functions of type $\Loc
\to \INT$.  We use $\Mem$ to refer to the set of all memories.
Evaluation is left-to-right, and conjunction is deliberately
non-short-circuiting.  Values and $\mathsf{skip}$ have no outgoing
transition. The relations are deterministic, and every non-terminal
configuration can progress. The definitions are standard, and we
give the most relevant rules next.

\begin{RULES}
  \ZEROPREMISERULE{\langle x,m\rangle
  \AStepL{\lread(x)}
  \langle m(x),m\rangle}
\quad
\ONEPREMISERULE{
  \langle e,m\rangle
  \AStepL{\ell}
  \langle e',m'\rangle
}{
  \langle x:=e,m\rangle
  \PStepL{\ell}
  \langle x:=e',m'\rangle
}
\quad
\ZEROPREMISERULE{
  \langle x:=v,m\rangle
  \PStepL{\lwrite(x)}
  \langle \mathsf{skip},m[x\mapsto v]\rangle
}
\\\\
\ZEROPREMISERULE{
  \left\langle
    \mathsf{while}\ b\ \mathsf{do}\ P,
    m
  \right\rangle
  \PStepL{\tau}
  \left\langle
    \mathsf{if}\ b\ \mathsf{then}\
      \bigl(P;\mathsf{while}\ b\ \mathsf{do}\ P\bigr)\
    \mathsf{else}\ \mathsf{skip},
    m
  \right\rangle
}
\end{RULES}

\NI The full rule sets and this basic metatheory are standard and
given in Appendix \ref{sec:small-step-ordinary}. We write $\PStepT{s}$
for the reflexive and transitive closure of $\PStepL{l}$. The length
$|s|$ counts every transition, expression steps included; the
\emph{access trace} $\mathsf{acc}(s)$ is the subsequence of non-$\tau$
labels, whose length is the number of memory accesses the run
performs.  A program \emph{terminates} with final memory $m'$, written
$\langle P,m\rangle\Downarrow m'$, when $\langle
P,m\rangle\PStepT{s}\langle\mathsf{skip},m'\rangle$ for some $s$; it
\emph{diverges}, written $\langle P,m\rangle\Uparrow$, when for every
$n$ some $s$ with $|s|=n$ extends its run.  The transitions give a
partial function on configurations $\langle P,m\rangle$. It will later
be convenient to extend this to total functions. We do this by
defining a deterministic \emph{absorbing run} denoted $\RunDT{n}(C)$: iterate
the step function $n$ times, letting terminal configurations step to
themselves with label $\tau$, and accumulate the labels.  Its result
is the unique configuration-with-trace reached after $n$ steps.
 \section{A probabilistic model of Rowhammer}\label{section_probabilistic_model_of_rowhammer}

This section addresses \RQONE, \RQTWO\ and \RQTHREE: we now refine the ordinary
semantics from the previous section with probabilistic Rowhammer fault models. We aim for the
interface between the programming language semantics and the fault
model to be minimal and canonical.  The syntax, evaluation order, and
reduction contexts are unchanged, and the refinement touches exactly
one place: the $\lread(x)$- and $\lwrite(x)$-labelled transitions of
\Cref{section_language} additionally sample from the Rowhammer fault model
and produce a probability distribution (\EG bit-flips in successor memories), while
$\tau$-labelled transitions remain deterministic.

With the changes to the read and write functions already explained in the Introduction,
this amounts to changing the state transformer semantics of programs
from type
\[
   \Mem \to \Mem
\]
from \Cref{section_language} into probabilistic state-transformers
\[
\Mem \to \FDist(\Mem)
\]
(and likewise for expressions).  The fault probabilities are lifted in
a canonical way to general programs using the probability monad from
\Cref{section_mathematical_preliminaries}.

\begin{FIGURE}
\IMAGE{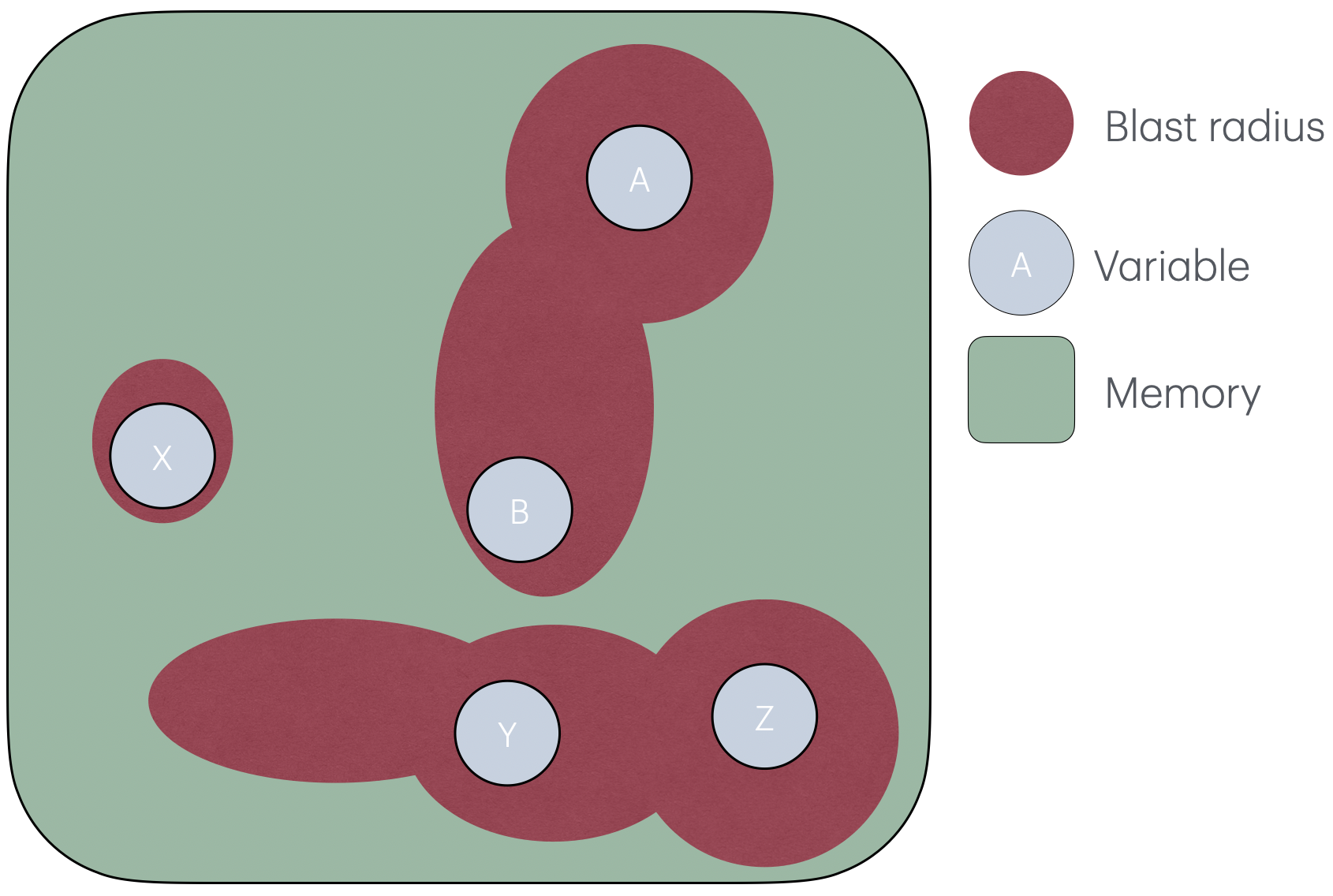}
\caption{The idea behind the fault model is that each variable $x$
  has a 'blast radius', a set of other memory locations that may get
  flipped when $x$ is accessed. Locations outside the blast radius
  will not be flipped. The fault model formalises the shape of each
  blast radius, and the flip probabilities.  Unlike Figures \ref{fig:blast-radius} and \ref{fig:separation}, this image does not show
  row shapes, because our formal model does not assume any specific geometry
  of blast radii.}\label{fig_blast_radius}
\end{FIGURE}

\paragraph{Why reads return a joint distribution.}
There is a small but important typing choice in the Rowhammer semantics.
The ordinary read operation has type
\[
\lread : (\Loc \times \Mem) \to \Val
\]
Since a Rowhammer access to a location $x$ is intended to disturb only
other locations, one might expect the probabilistic read operation to
have type
\[
(\Loc \times \Mem) \to (\Val \times \FDist(\Mem))
\]
returning the value stored at $x$ together with a distribution over
successor memories.  In the model used in this paper, this would
indeed be sufficient for a \EMPH{single} read step: the value returned by the
read is determined by the input memory, and the probabilistic fault only
affects the successor memory.  We instead use the slightly more general
type
\[
\lread : (\Loc \times \Mem) \to \FDist(\Val \times \Mem)
\]
This is the standard form of a probabilistic state transformer with a
return value: it produces a joint distribution over the observed
result and the successor state.  The more specialised type $\Val
\times \FDist(\Mem)$ embeds into this one by pairing every successor
memory with the same returned value.  Using the joint form
$\FDist(\Val \times \Mem)$ avoids baking this special case into the
semantic interface and ensures that read operations share the same
monadic shape as the rest of the semantics. This leaves the
abstraction general enough to model bit-flip errors, where a
physical fault correlates the specific value returned by a noisy read
with the resulting successor memory, without changing the semantic
type of reads.

\subsection{The fault-kernel contract}
\label{sec:fault-kernel-contract}

Fault kernels is our name for the  probability kernels that model Rowhammer. The idea behind fault kernels is
simple: they describe \EMPH{what} memory cells may be affected by
Rowhammer effects when a location $x$ is accessed and \EMPH{how}.
Figure \ref{fig_blast_radius} visualises this intuition.  We formalise
this intuition with the following predicates.
\begin{itemize}

\item $\adj(br,locs,x)$ gives the set of all memory locations within the
  blast radius $br$ of location $x$.

\item $\pflip(S,m) \in \FDist(\Mem)$ is a probability distribution
  that is allowed to flip bits in the finite set $S$ of locations in
  memory $m$. Locations outside of $S$ must remain unchanged.

\end{itemize}
We can now state the modelling assumptions we will use in
\Cref{section_soundness} for proving the soundness of physical
separation as a Rowhammer defence.

\begin{itemize}

\item \textbf{(F1) Access-triggered.}  Faults are sampled exactly at the
  $\lread(x)$- and $\lwrite(x)$-labelled transitions of
  \Cref{section_language}, $\tau$-transitions are fault-free, modelled
  as the Dirac distribution.  (F1) is enforced by construction of the
  rules below, not assumed.

\item \textbf{(F2) Spatial confinement (frame-locality).}  Every memory in the
  support of $\pflip(S,m)$ agrees with $m$ outside the victim set $S$.
  Within $S$ the corruption may assign arbitrary probabilities to
  arbitrary, correlated corruptions: the kernel may behave
  adversarially, subject only to this clause.

\item \textbf{(F3) History-freeness.}  The kernel depends only on the current
  victim set and the current memory: $\pflip$ is a function
  $\POWERSETFIN(\Loc)\to\Mem\to\FDist(\Mem)$, so a fault
  cannot depend on the access history, on time, or on hidden hardware
  state.  This is a genuine modelling restriction. It excludes, for
  example, charge accumulation across repeated activations and
  refresh-cycle effects, generalising to stateful kernels is future
  work, as discussed in \Cref{section_conclusion}.

\item \textbf{(F4) Stationarity.}  The same kernel governs every access of a
  run, so fault behaviour does not drift over time.  Like (F3), this
  is encoded by $\pflip$ being a single function rather than assumed
  as a separate axiom.

\end{itemize}

\NI Clauses (F1-F4), considered in isolation, do not require the
location whose access triggers a fault to lie outside its own victim
set: the abstract contract permits $x\in\adj(br,locs,x)$.  This
generality is not needed for Rowhammer which does not allow
``self-hammering'': the aggressor location cannot itself be a
victim.  Nevertheless, no additional axiom is required for the
executions covered by our main theorem.  Program well-formedness gives
$x\in locs$ for every program access, and safe physical separation
then gives $ x\notin\adj(br,locs,x) $ by instantiating
protected--victim disjointness with the protected location $x$ itself.
Clause (F2) consequently ensures that the fault preserves the accessed
location.  Thus, the abstract contract is slightly more general than
the intended DRAM interpretation, while the non-self-disturbance
property required for protected executions is a derived consequence of
the theorem hypotheses rather than a separate kernel axiom.

\subsection{Physical adequacy and scope of our fault model}

Our semantics is not a physically faithful device-level model of the
electrical mechanisms underlying Rowhammer, nor is it intended to
predict the probability of a bit-flip on a particular DRAM module.
Our abstraction can be read as a conservative attacker model for
Rowhammer-style disturbance: the bit-flips admissible by our model can
be more extreme than those observed in physical DRAM.  Therefore, a
universal theorem proved for every admissible abstract kernel applies
a fortiori to any more constrained concrete device model that refines
it.

\subsection{The fault model, mechanised}
\label{sec:fault-model-lean}

The Lean development renders the contract in three declarations.
Separation is an abstract structure: a blast-radius function, a safety
predicate, and the single fact the soundness proof uses about
them: with physical separation, the victim set of an access to a protected
location is disjoint from the protected set.

\begin{lstlisting}[language=Lean]
structure SeparationModel where
  adj : Nat → LocSet → Loc → LocSet
  isSafe : Nat → LocSet → Prop
  /-- Safe layouts: the victim set of an access to a protected row is
      disjoint from the protected set. -/
  safe_disjoint : ∀ {br : Nat} {locs : LocSet} {x : Loc},
    isSafe br locs → x ∈ locs → ∀ y ∈ locs, y ∉ adj br locs x
\end{lstlisting}
 
Our abstract fault contract does not build the row structure of DRAM
directly into the programming language semantics.  Instead, a
separation model supplies an abstract victim-set function
$\adj(br,locs,x)$, identifying the locations that may be corrupted by
an access to $x$.  The theory does not require this function to be
derived from a metric, or even to be symmetric.  This is both simpler
and more general than fixing a particular DRAM organisation.

The mechanisation includes a concrete distance-based instance in which
locations are natural numbers and the victim set is determined by
numerical distance.  This instance is primarily a simple witness that
the abstract contract is non-vacuous: numerical address distance
should not itself be identified with physical DRAM distance unless
locations have already been interpreted as physical row identifiers.
A row-based model is a special case of the abstract interface:
locations may be taken to denote rows directly, or the victim-set
function may be defined through a mapping from program-visible
locations to DRAM banks and rows.

A fault kernel is a function from victim sets and memories to
distributions over memories.  The \emph{type} already encodes (F3) and
(F4): the kernel sees neither history nor time.

\begin{lstlisting}[language=Lean]
abbrev ProbRHFlip (_sep : SeparationModel) :=
  LocSet → Memory → Dist Memory
\end{lstlisting}
 
\noindent
Clause (F2) is the sole behavioural assumption, stated as a typeclass:

\begin{lstlisting}[language=Lean]
class HasProbFlipFrame (sep : SeparationModel) (f : ProbRHFlip sep) : Prop where
  outside_unchanged : ∀ (S : LocSet) (m m' : Memory) (x : Loc),
    m' ∈ (f S m).support → x ∉ S → m' x = m x
\end{lstlisting}
 
\noindent
Clause (F1) is enforced by construction of the transition rules below.

The contract is not vacuous, and not satisfied only by fault-free
kernels.  The identity kernel \texttt{noFlip} (every fault is a no-op)
satisfies it and embeds the ordinary semantics
(\Cref{sec:goal-formal}). At the other extreme, so does the worst-case
kernel that deterministically corrupts the \emph{entire} victim set at
every access:

\begin{lstlisting}[language=Lean]
noncomputable def corruptAll (sep : SeparationModel) : ProbRHFlip sep :=
  fun S m => dirac (fun x => if x ∈ S then m x + 1 else m x)
\end{lstlisting}
 
\noindent
The mechanisation proves \texttt{corruptAll} frame-local and genuinely
faulting (every outcome differs from the pre-fault memory at every
victim location), and exhibits concrete physical separation of the distance
model together with a well-formed program over it, so every hypothesis
of the goal theorem is satisfiable simultaneously with a kernel that
really flips bits.

\subsection{Probabilistic transitions}
\label{sec:prob-transitions}

The transition judgements
\[
  \langle e,m\rangle \RHAStep \mu,
  \qquad
  \langle b,m\rangle \RHBStep \mu,
  \qquad
  \langle P,m\rangle \RHPStep \mu
\]
return a \emph{distribution} $\mu$ over successor configurations rather
than a single successor, and are otherwise the rules of
\Cref{section_language} with the same left-to-right, non-short-circuit
evaluation order.  Exactly two rules sample the kernel, matching (F1).
A variable read first obtains the value $m(x)$ and then samples the
fault triggered by that access. The sampled fault therefore cannot
retroactively change the value returned by the current read, but it may
change the memory seen by later reads:
\[
\frac{\strut}{
  \langle x,m\rangle
  \RHAStep
  \left(
    \mathbf{let}\ m'\leftarrow
      \pflip\bigl(\adj(br,locs,x),m\bigr)
    \ \mathbf{in}\
    \dirac\bigl(\langle m(x),m'\rangle\bigr)
  \right)
}
\tag{\textsc{RH-A-Read}}
\]
Dually, once the right-hand side of an assignment is a value, the
assignment writes it and samples the fault triggered by the write, on
the \emph{updated} memory:
\[
\frac{\strut}{
  \langle x:=v,m\rangle
  \RHPStep
  \left(
    \mathbf{let}\ m'\leftarrow
      \pflip\bigl(
        \adj(br,locs,x),
        m[x\mapsto v]
      \bigr)
    \ \mathbf{in}\
    \dirac\bigl(\langle\mathsf{skip},m'\rangle\bigr)
  \right)
}
\tag{\textsc{RH-P-Assign}}
\]

\NI Under the well-formedness and physical separation hypotheses of the main
theorem, the accessed location lies outside its victim set.  Hence
frame-locality preserves the stored value at a read location and
preserves the value just installed by an assignment.  This fact is
derived from physical separation. It is not an additional assumption
on the fault kernel.  Every other rule is the Dirac or functorial
lifting of its deterministic counterpart. The full rule sets for all
three judgements are in \Cref{sec:small-step-faulty}.  There is no
transition from $\langle\mathsf{skip},m\rangle$.  The rules are
syntax-directed, so every non-terminal program configuration
determines a unique one-step distribution: the fault randomises the
memory contents, never which access happens next.  The semantics is
therefore probabilistic but not nondeterministic.  We write
$\RHStep(\langle P,m\rangle)$ for the program-level distribution.

\subsection{Finite probabilistic executions}
\label{sec:finite-prob-executions}

For iteration, define an absorbing extension of the one-step kernel by
\begin{itemize}

  \item
$  \RHStepHat\bigl(\langle\mathsf{skip},m\rangle\bigr)
  =
  \dirac\bigl(\langle\mathsf{skip},m\rangle\bigr)$,

\item  $\RHStepHat\bigl(\langle P,m\rangle\bigr)
  =
  \RHStep\bigl(\langle P,m\rangle\bigr)$ assuming  that $P\neq\mathsf{skip}$,

\end{itemize}
and the distribution after exactly $n$ applications of the absorbing
kernel by
\begin{itemize}

  \item $\RunRH{0}(C)=\dirac(C)$,

  \item $ \RunRH{n+1}(C) = \mathbf{let}\ C'\leftarrow\RHStepHat(C)
    \ \mathbf{in}\ \RunRH{n}(C')$.

\end{itemize}
The absorbing extension exists only to define finite runs. The
operational semantics still has no transition from a terminal
configuration.  Because expression reductions are lifted through
assignments and conditionals one step at a time, the index $n$ counts
genuine small-step transitions. A read or write together with its
immediately triggered fault sample counts as one probabilistic
transition.

In the mechanisation the absorbing kernel is a total \emph{function}
(\texttt{stepRHHat}), and the displayed judgement rules are inductive
relations, that the two agree on non-terminal configurations (the
kernel computes exactly the judgement's distribution, and the
judgement's distribution is unique) is proved
(\texttt{pRHStep\_iff\_stepRHHat}), so results stated over the runs
below are statements about the displayed rules.

Because the rules are syntax-directed, \emph{which} access a
configuration performs next is a function of the configuration alone.
The $n$-step run therefore extends canonically to a trace-carrying run
$\RunRHT{n}$ that accumulates the label sequence of
\Cref{section_language} alongside the configuration, assigning to every
outcome its access trace. Absorbing steps at terminal configurations
contribute $\tau$, which the access trace erases.  The goal theorem's
trace agreement (\Cref{sec:goal-formal}) is stated over this run.

\subsection{Mixed termination}
\label{sec:mixed-termination}

Without mitigations, Rowhammer faults may change not only final values
but also whether a program terminates.  Indeed, a single initial
configuration may induce a probabilistic mixture of terminating and
non-terminating executions.  Consider the following program:
\[
   \mathsf{while}\ x>17\ \mathsf{do} \bigl(y:=z+1; x:=x-1\bigr).
\]
In the ordinary semantics, every iteration decreases $x$ by one, so
the program terminates.  Now suppose a Rowhammer bit-flip triggered by
writing $y$ may corrupt $x$.  As a concrete illustration, consider a
kernel that leaves $x$ unchanged with probability $1-p$, changes $x$
to $x+2$ with probability $0 < p < 1$ after the write to $y$, and
otherwise leaves the locations relevant to the example unchanged.
Hence the program terminates with positive probability and runs
forever with positive probability.  This example is deliberately
schematic: it demonstrates behaviour permitted by the abstract fault
contract, not a quantitative claim about a particular DRAM device.

For sufficiently many execution steps, some probability mass may
already be on configurations $\langle\mathsf{skip},m'\rangle$, while
the remaining mass is on non-terminal configurations. The exact
probabilities of non-termination are available only as limits.  We do
not mechanise these limiting probabilities in the present development.
The present finite step probability-mass-function approach is
sufficient for the theorem we establish in later sections. An explicit
account of divergence is possible, but would require additional
semantic infrastructure.

\subsection{Language independence}

The language-parametricity of the model comes from the small semantic
signature required by fault kernels.  To instantiate the interface
for a language, one must identify memories, locations, and the semantic
clauses that perform memory actions.  The ordinary clauses for those
actions are then factored into an intended logical access followed by a
call to the fault kernel determined by the accessed location.  All
other language constructs interact with Rowhammer only through the
memory actions they eventually perform.  Consequently, the interface can
be applied directly to explicitly imperative languages and indirectly to
non-imperative languages by applying it to an intermediate language,
runtime semantics, or machine semantics in which memory actions are
explicit.

The language-independence claim is also supported by recent work on
reusable architecture semantics.  ArchSem \cite{PeramiT:archsem}, for
example, factors instruction set architecture semantics through a
small algebraic-effect interface between ISA descriptions and memory
models.  In that interface, instruction semantics does not directly
manipulate a monolithic global machine state.  Instead, it emits
outcomes such as memory reads and writes, barriers, cache operations,
TLB operations, and exceptions.  The memory-relevant part of the
interface is particularly close to what our model requires: a memory
request records an access kind, address, address space, size, and tag
information, and the outcome interface contains distinguished
memory-read and memory-write effects.  Our Rowhammer interface can be
understood as an analogous factorisation, but at a different semantic
boundary.  ArchSem factors ISA behaviour from the memory and
concurrency model. We factor the physical fault mechanism from the
programming language or machine semantics.  To instantiate our model
for a new language or architecture, one need not rebuild the Rowhammer
theory around that language's syntax.  It suffices to identify the
semantic events that correspond to memory accesses, compute the set of
physical locations exposed by each such access, and interpret the
access through the probabilistic fault kernel.  In a WHILE
language these events are source-level variable reads and writes, but
in an architecture semantics such as ArchSem they would correspond
instead to memory outcomes such as \texttt{MemRead} and
\texttt{MemWrite}.  This is the precise sense in which the interface
is substantially independent of the target programming language: the
language-specific part is the production of memory-access events,
while the Rowhammer-specific part is a reusable interpretation of
those events as possible probabilistic disturbances of physical
memory.
 \section{Soundness of physical separation}\label{section_soundness}

This section addresses our \RQFOUR\ and proves the paper's first
main theorem: under the fault kernel contract of
\Cref{sec:fault-kernel-contract}, physical separation eliminates
Rowhammer faults from being observable in the results of program execution.
We first state the goal informally, then
define the projection being compared, give the formal statement
together with its mechanised form, walk through every hypothesis, and
sketch the proof.

The idea is simple: given a program $P$ such that all variables $locs$
occurring in $P$ are physically separated, meaning that they are all
outside each other's blast radius.  Then the probabilistic semantics of
$\langle P, m \rangle$ when restricted to $locs$ is the Dirac
distribution of the deterministic run of $\langle P, m \rangle$.  All
formal ingredients are already defined, except comparing a
deterministic with a probabilistic run of $\langle P,m\rangle$.

\subsection{The protected view}
\label{sec:protected-view}

We call any element of $(\mathsf{Prog}\times\Mem)\times \ACCESS^{*}$
a \EMPH{configuration with labels}.
Write $\RunDT{n}(C)\in(\mathsf{Prog}\times\Mem)\times \ACCESS^{*}$ for
the deterministic $n$-step absorbing run with its accumulated label
sequence (terminal configurations step to themselves, contributing
$\tau$, \Cref{section_language}), and
$\RunRHT{n}(C)\in\FDist\bigl((\mathsf{Prog}\times\Mem)\times
\ACCESS^{*}\bigr)$ for the trace-carrying Rowhammer run of
\Cref{sec:finite-prob-executions}.  The \emph{protected view} of a
configuration with labels is
\[
  \protV\bigl(\langle P,m\rangle, s\bigr)
  \;\triangleq\;
  \bigl(P,\ m{\restriction}locs,\ \mathsf{acc}(s)\bigr),
\]
the residual program, the restriction of the memory to the protected
locations, and the access trace.  In Lean, the restriction is a total
map on the subtype of protected locations, and the access trace erases
$\tau$:

\begin{lstlisting}[language=Lean]
def protView (locs : LocSet) :
    (Prog × Memory) × List Access →
      Prog × ({x : Loc // x ∈ locs} → Val) × List Access :=
  fun r => (r.1.1, fun x => r.1.2 x.1, accessTrace r.2)
\end{lstlisting}
 
\noindent
These three components are exactly what the program's execution can
depend on.  What the projection discards, unprotected DRAM subject to
Rowhammer outside $locs$, is exactly what a physical separation
defence sacrifices to the fault, and a well-formed program never reads
it.

\subsection{Formal statement}
\label{sec:goal-formal}

We are now in a position to state the soundness theorem, both
informally and with extracted Lean.

\begin{theorem}[Soundness of physical separation, formal]
\label{thm:goal-formal}\label{goal:physical-separation}
Let $\pflip$ satisfy the fault kernel contract
(\Cref{sec:fault-kernel-contract}), let $\safe(br,locs)$ hold, and let
$\wf_{locs}(P)$.  Then for every initial memory $m$ and every
$n\in\Nat$,
\[
  \FDist.map\ \protV\ \bigl(\RunRHT{n}(\langle P,m\rangle)\bigr)
  \;=\;
  \dirac\Bigl(\protV\bigl(\RunDT{n}(\langle P,m\rangle)\bigr)\Bigr).
\]
\end{theorem}

\noindent
Equivalently, by the support characterisation of the point mass
(\Cref{section_mathematical_preliminaries}): every outcome in the
support of the $n$-step Rowhammer run has the residual program, the
protected memory, and the access trace of the $n$-step ordinary run.
The mechanised statement quantifies in addition over every separation
geometry:

\begin{lstlisting}[language=Lean]
def PhysicalSeparationSound : Prop :=
  ∀ (sep : SeparationModel) (f : ProbRHFlip sep),
    HasProbFlipFrame sep f →
    ∀ (br : Nat) (locs : LocSet),
      sep.isSafe br locs →
      ∀ (p : Prog), p.WellFormed locs →
        ∀ (m : Memory) (n : Nat),
          Dist.mapD (protView locs)
              (runRHTrace sep f br locs n (p, m))
            = Dist.dirac (protView locs (runDTrace n (p, m)))
\end{lstlisting}
 
\NI It might be of interest to see that all hypotheses
in the mechanised statement are needed, and we walk through them in
order next.

\begin{itemize}

  \item \textbf{\texttt{sep : SeparationModel}.}  The theorem is universal
    over abstract separation geometries (\Cref{sec:fault-model-lean}):
    a blast-radius function $\adj$, a safety predicate $\safe$, and
    the disjointness fact \texttt{safe\_disjoint} connecting them.
    The proof uses nothing else about the geometry, no row
    arithmetic, no linear layout, so the result applies to any
    placement scheme whose safety condition implies the disjointness
    fact, with the distance model of \Cref{sec:fault-model-lean} as
    the motivating instance.

  \item \textbf{\texttt{f : ProbRHFlip sep}.}  The fault kernel.  Its type,
    $\mathsf{LocSet}\to\Mem\to\FDist(\Mem)$, encodes contract clauses
    (F3) and (F4): the kernel can consult only the current victim set
    and memory, and the same kernel governs every access.  Nothing
    about the distribution itself is fixed.

\item\textbf{\texttt{HasProbFlipFrame sep f}.}
Contract clause (F2), spatial confinement is \emph{the sole behavioural
hypothesis}.  Every memory in the support of $\pflip(S,m)$ agrees with
$m$ outside $S$.  Dropping it makes the theorem false immediately: a
kernel that flips a protected location on some access satisfies
everything else, and its Rowhammer run visibly diverges from the
ordinary one on protected memory.

\item \textbf{{\texttt{br} and \texttt{locs}.}}
The blast radius and the protected set.  Both are abstract parameters
threaded through $\adj$ where $locs$ is a \emph{finite} set of locations (a
\textsf{Finset}), the memory locations the program uses for computation.

\item \textbf{{\texttt{sep.isSafe br locs}.}}
Physical separation, the condition the defence establishes, \EG by
guard rows.  Through \texttt{safe\_disjoint} it yields the only fact the
proof needs: the victim set of an access to a protected location is
disjoint from the protected set.  Without it the theorem fails even for
kernels obeying the whole contract: accessing one protected location may
corrupt another protected location inside its blast radius
(\Cref{ex:dead-branch}).

\item \textbf{{\texttt{p.WellFormed locs}.}}
Every location the program reads or writes is in $locs$
(\Cref{section_language}).  This clause makes separation compose with
execution: it guarantees that every access the run performs is an
access to a \emph{protected} location, so \texttt{safe\_disjoint}
applies to every victim set the run generates.  Without it, an access
to some $z\notin locs$ triggers a fault on $\adj(br,locs,z)$, about
which $\safe(br,locs)$ says nothing, the blast radius of an unprotected
location may well intersect $locs$ (\Cref{ex:unowned-access}).

\item \textbf{{\texttt{m} and \texttt{n}.}}
The initial memory and the number of steps, both universally
quantified.  Quantifying over every finite $n$ is the relational-safety
reading of \Cref{goal:physical-separation}: it covers terminating and
diverging programs alike, with no limit construction.

\end{itemize}

\NI It may also be interesting to see what is \emph{not} assumed:
no bound on flip probabilities (faults may be certain), no
independence between locations or between fault events, no
unbiasedness, no particular distribution.  Within the blast radius the
kernel may be adversarial.  The guarantee is distribution-free: it
follows from spatial confinement and separation alone.

\subsection{Corollaries}
\label{sec:corollaries}

Two clauses of \Cref{goal:physical-separation} are worth recording as
standalone consequences of the point mass equation.

\begin{corollary}[Exact termination correspondence]
\label{cor:goal-termination}
Under the hypotheses of \Cref{thm:goal-formal}, if
$\langle P,m\rangle\PStepT{s}\langle\mathsf{skip},m'\rangle$, then every
configuration in the support of $\RunRH{|s|}(\langle P,m\rangle)$ is
terminal and agrees with $m'$ on $locs$: the Rowhammer execution
terminates after exactly the same number of abstract-machine steps, with
the same protected final memory (and, by \Cref{thm:goal-formal}, the same
access trace).
\end{corollary}

\begin{lstlisting}[language=Lean]
def TerminationCorrespondence : Prop :=
  ∀ (sep : SeparationModel) (f : ProbRHFlip sep),
    HasProbFlipFrame sep f →
    ∀ (br : Nat) (locs : LocSet),
      sep.isSafe br locs →
      ∀ (p : Prog), p.WellFormed locs →
        ∀ (m m' : Memory) (s : List Access),
          PTrace (p, m) s (Prog.skip, m') →
          ∀ c ∈ (runRH sep f br locs s.length (p, m)).support,
            c.1 = Prog.skip ∧ ∀ x ∈ locs, c.2 x = m' x
\end{lstlisting}
 
\begin{corollary}[Divergence preservation]
\label{cor:goal-divergence}
Under the hypotheses of \Cref{thm:goal-formal}, if
$\langle P,m\rangle\Uparrow$, then for every $n$ every configuration in
the support of $\RunRH{n}(\langle P,m\rangle)$ is non-terminal (and agrees
with the ordinary run on the protected state).
\end{corollary}

\begin{lstlisting}[language=Lean]
def DivergencePreservation : Prop :=
  ∀ (sep : SeparationModel) (f : ProbRHFlip sep),
    HasProbFlipFrame sep f →
    ∀ (br : Nat) (locs : LocSet),
      sep.isSafe br locs →
      ∀ (p : Prog), p.WellFormed locs →
        ∀ (m : Memory), Diverges p m →
          ∀ (n : Nat),
            ∀ c ∈ (runRH sep f br locs n (p, m)).support,
              c.1 ≠ Prog.skip
\end{lstlisting}
 
\subsection{Proof sketch}
\label{sec:proof-sketch}

The proof is a probabilistic simulation whose invariant is
\emph{parametric in the preserved region}: a predicate $\keep$ on
locations, instantiated only at the very end with membership in $locs$.
It proceeds in three main steps.

\emph{One-step invariant} (\texttt{stepRHHat\_agree}).  Suppose the
Rowhammer and ordinary configurations share the same residual program
and their memories agree on $\keep$, and suppose (i) every location the
program uses satisfies $\keep$, and (ii) every fault sample triggered by
an access the program can perform preserves $\keep$-locations.  Then
every outcome of one application of the Rowhammer kernel again shares
the residual program of the ordinary step and agrees with it on
$\keep$.  The residual program is preserved exactly, not just up to
$\keep$, because reads return the same values on $\keep$-agreeing
memories, so both machines make the same control decisions.

\emph{Iteration} (\texttt{runRHTrace\_protected}).  The one-step
invariant is iterated over the step count $n$.  Because the rules are
syntax-directed, the label of the next transition is a function of the
residual program alone. Preserving the residual program exactly
therefore forces both runs to emit the \emph{same label sequence}, which
gives the access-trace clause for free.  Well-formedness is preserved
along the run, so clauses (i) and (ii) remain available at every step.

\emph{Instantiation and conversion.}  Take $\keep(\ell)\triangleq\ell\in
locs$.  Clause (i) is program well-formedness.  Clause (ii) is exactly
frame-locality (F2) composed with \texttt{safe\_disjoint}: an access to
$x\in locs$ triggers a fault confined to $\adj(br,locs,x)$, which is
disjoint from $locs$, so every sample fixes every protected location.
The invariant then says that every outcome in the support of
$\RunRHT{n}$ has the protected view of $\RunDT{n}$. The support
characterisation of \Cref{section_mathematical_preliminaries} converts
this into the point mass equation.  The corollaries follow from the
central equation together with determinism, progress, and the trace
realisation lemma of the ordinary semantics
(\Cref{sec:small-step-ordinary}). The no-fault embedding is by
computation.

We can now show that the deterministic semantics is a special case of
the probabilistic semantics. This gives us confidence in the
correctness of the probabilistic semantics.

\begin{theorem}[No-fault embedding]
\label{thm:no-fault-finite-horizon}
Let $sep$ be any separation model, let $br$ be any blast-radius
parameter, let $locs$ be any protected set, and let $c$ be any initial
configuration.  If the Rowhammer fault kernel is the identity kernel
$\NoFlip(sep)$, then $n$ steps with the Rowhammer semantics is
exactly the Dirac distribution concentrated on the $n$-step
deterministic trace run: This requires neither safe physical
separation nor program well-formedness.
\end{theorem}
This theorem is mechanised in Lean as \texttt{runRHTrace\_noFlip}.
 \providecommand{\Low}{\ensuremath{\mathsf{Low}}}
\providecommand{\High}{\ensuremath{\mathsf{High}}}

\section{Information flow security under Rowhammer}
\label{sec:rowhammer-information-flow-intuition}

We now address our \RQFIVE\ whether the well-developed and widely used
theory of information flow
\cite{goguen1982security,sabelfeld2003language} can be lifted to our
Rowhammer model.  We do this by defining \EMPH{relative
  non-interference} and prove that physical separation guarantees
relative non-interference.  Relative non-interference means that any
information leak that happens with Rowhammer also happens without
Rowhammer\footnote{\STARTCF relative consistency results like
G\"odel's proof that classical logic does not introduce unsoundness
not already present in intuitionistic logic \cite{Godel1933eDoubleNegation},
or the G\"odel / Cohen results that ZFC set theory is relatively
consistent over ZF
\cite{goedel1940consistency,cohen1963independence}.}.
This is of interest, as Rowhammer is usually presented as an integrity
failure: an access to one physical row may corrupt data stored in
another.  But the same mechanism can also create a confidentiality
failure:  a Rowhammer induced flip in a high location may alter a branch condition,
loop condition, or value that is later copied to low memory. Moreover,
a fault in a low location may directly change an attacker-visible
result.  The theory of information flow is tailor-made for studying
such phenomena, and provides a natural test of whether physical
separation removes the semantic consequences of Rowhammer rather than
merely making faults less likely.

\paragraph{Low equivalence and non-interference.}
In order to explain relative non-interference in detail, we need to sketch the key ideas
behind information flow security, see
\cite{goguen1982security,sabelfeld2003language} for details.  Fix a
policy that classifies every memory location as either \Low{} or
\High{}.  Two memories are \emph{low equivalent} when they agree at
every \Low{} location (they  may differ arbitrarily at \High{} locations).
A program is \EMPH{non-interfering} when executions from
low-equivalent initial memories end with memories that are also
low-equivalent (we discuss termination later).  Non-interference is
therefore a \EMPH{hyperproperty} \cite{clarkson2010hyperproperties}:
it relates pairs of executions rather than inspecting one execution in
isolation. The ordinary semantics is deterministic, so
indistinguishability means equality of the two observations.  The
Rowhammer semantics maps an initial memory to a probability
distribution over executions.  Probability-sensitive non-interference
consequently requires equality of the induced distributions over low
observations.

\paragraph{What physical separation can and cannot guarantee.}
The physical separation discussed in previous sections
must protect every location used by the program, not only
its low locations.  Protecting only low variables is insufficient: corrupting
a high variable may change control flow and thereby alter a later low write,
low access, or termination behaviour.  The separation condition used in this
paper therefore protects the complete set of locations over which the
program is well formed.  That means we can only give a
\emph{no-new-leaks} guarantee: physical separation does not make a
program secure that is insecure in the deterministic
semantics! For example,
\[
  L := H
\]
explicitly leaks a high value into a low variable and is therefore
interfering in both semantics.  Physical separation only ensures that
the Rowhammer semantics has exactly the same protected observations as
the ordinary semantics.  Hence, with physical separation, a Rowhammer
information flow counterexample is only a counterexample that was
already present without faults.
Without separation this reasoning fails.  If an access to a program's
location can disturb another location, a frame-local fault kernel may
change a low variable or a high value controlling future low behaviour.  The
fault kernel may assign arbitrary probabilities to such corruptions inside the
victim set.  Safe separation rules out this channel structurally: the victim
set generated by an access to a protected location is disjoint from every
protected location.
 \section{Mechanised observation-parametric non-interference}
\label{sec:rowhammer-noninterference-formalism}

This section formalises and proves our relative non-interference
results.  The Rowhammer semantics of previous sections remains
unchanged.  Our central design choice is that every observer consumes
the protected view already covered by the physical separation theorem.
Thus the security results are consequences of semantic collapse proven
already, rather than a second operational simulation argument. This
gives us additional confidence in our mathematical rendering of
Rowhammer as a programming language feature.

\subsection{Formal development}
\paragraph{Observation-parametric security.}
There is no single universally appropriate notion of low observation (\EG is termination
low observable).  The mechanised
development therefore parameterises non-interference by a function on the
existing protected view, which contains the residual program, protected
memory, and access trace.  We define four concrete observers:
\begin{itemize}

\item Protected low memory only

\item The address trace of reads and writes to low locations

\item Terminal status, protected low memory, and the low access trace

\item Protected low memory and the low access trace, but not terminal
  status
\end{itemize}
The third observer is progress-sensitive because observations are compared at
every finite number of steps: it can reveal the exact abstract-machine step at which
the residual program first becomes \texttt{skip}.  Omitting the terminal-status
bit does not by itself yield classical termination-insensitive
non-interference, since low memory or low accesses observed at successive
steps may still reveal progress.  The mechanisation therefore makes no
claim that its fourth observer is a complete formalisation of classical
termination-insensitive non-interference.

\paragraph{Security policies and low equivalence, protected views and observations.}
Let $ \SecLevel \triangleq \{\Low,\High\} $ and let a security policy
be a function $ \gamma:\Loc\to\SecLevel$.  The low and high location
subtypes are $ L_\gamma \triangleq \{x:\Loc\mid\gamma(x)=\Low\}, $ and
$ H_\gamma \triangleq \{x:\Loc\mid\gamma(x)=\High\}$. For $m\in\Mem$,
the two restricted views are
\begin{itemize}

\item $\LowView_\gamma(m)(x) \triangleq m(x)$ assuming $(x\in
  L_\gamma)$

\item $\HighView_\gamma(m)(x) \triangleq m(x)$ assuming $x\in
  H_\gamma$.

\end{itemize}
Two memories are low equivalent when
\[
  m_1\LowEqRel{\gamma}m_2
  \quad\Longleftrightarrow\quad
  \forall x:\Loc.\;
  \gamma(x)=\Low\Longrightarrow m_1(x)=m_2(x).
\]
The mechanisation proves that this pointwise definition is equivalent to
\[
  \LowView_\gamma(m_1)=\LowView_\gamma(m_2),
\]
and proves reflexivity, symmetry, and transitivity.
Let $\mathsf{Access}$ contain labels $\mathsf{read}(x)$,
$\mathsf{write}(x)$, and $\tau$.  The function $\AccessTrace$ filters out
$\tau$.  For the finite protected set $locs$, define
\[
  \PView(locs)
  \triangleq
  \mathsf{Prog}
  \times
  \bigl(\{x:\Loc\mid x\in locs\}\to\Val\bigr)
  \times
  \mathsf{List}(\mathsf{Access}).
\]
The protected projection is
\[
  \ProtView_{locs}((Q,m),s)
  \triangleq
  \bigl(Q,\,m|_{locs},\,\AccessTrace(s)\bigr).
\]
It retains exactly the three components covered by physical semantic
collapse: residual syntax, protected memory, and the access trace.
The deterministic protected view at step count $n$ is
\[
  \DProtectedView{n}_{locs}(P,m)
  \triangleq
  \ProtView_{locs}\bigl(\RunDTrace{n}(P,m)\bigr).
\]
The Rowhammer protected-view distribution is
\[
\begin{split}
  \RHProtectedView{n}_{sep,f,br,locs}(P,m)
  \triangleq
  \FDist.\mathsf{map}\;\ProtView_{locs}
  \bigl(\RunRHTrace{n}_{sep,f,br,locs}(P,m)\bigr).
\end{split}
\]

\paragraph{Non-interference definitions.}
Recall that we do not hard-code just one specific notion of how the
program execution is observed. Instead, let $\Omega$ be an observation
type.  A protected observer is a function
\[
  O:\PView(locs)\to\Omega.
\]
The two observations used in the definitions below are
\[
  \ObsOrd{n}_{locs,O}(P,m)
  \triangleq
  O\bigl(\DProtectedView{n}_{locs}(P,m)\bigr)
\]
and
\[
\begin{split}
  \ObsRH{n}_{sep,f,br,locs,O}(P,m)
  \triangleq
  \FDist.\mathsf{map}\;O
  \bigl(\RHProtectedView{n}_{sep,f,br,locs}(P,m)\bigr).
\end{split}
\]
See the Lean definitions \texttt{rowhammerProtectedView} and
\texttt{rowhammerObservation} that match this two-stage definition.
We can now define our notion of observer parametric non-interference.

\begin{definition}[Ordinary step bound non-interference]
\label{def:ordinary-observation-ni}
The predicate $\OrdNI(\gamma,locs,O,P)$ holds when
\[
\begin{split}
  \forall n\in\Nat.\;\forall m_1,m_2\in\Mem.\quad
  m_1\LowEqRel{\gamma}m_2
  \Longrightarrow
  \ObsOrd{n}_{locs,O}(P,m_1)
  =
  \ObsOrd{n}_{locs,O}(P,m_2).
\end{split}
\]
This is the Lean definition \texttt{OrdinaryNI}.
\end{definition}

\begin{definition}[Rowhammer step bound non-interference]
\label{def:rowhammer-observation-ni}
For a fixed separation model $sep$, fault kernel $f$, blast radius $br$, and
protected set $locs$, the predicate
$\RHNI(\gamma,sep,f,br,locs,O,P)$ holds when
\[
\begin{split}
  \forall n\in\Nat.\;\forall m_1,m_2\in\Mem.\quad
  m_1\LowEqRel{\gamma}m_2
  \Longrightarrow
  \ObsRH{n}_{sep,f,br,locs,O}(P,m_1)
  =
  \ObsRH{n}_{sep,f,br,locs,O}(P,m_2).
\end{split}
\]
The equality is equality of complete finite distributions.  The same fixed
kernel $f$ is used in both runs.  This is the Lean definition
\texttt{RowhammerNI}.
\end{definition}

\begin{definition}[Robust Rowhammer non-interference]
\label{def:robust-rowhammer-ni}
The predicate $\RobustRHNI(\gamma,sep,br,locs,O,P)$ holds when
\[
  \forall f.\;
  \mathsf{HasProbFlipFrame}(sep,f)
  \Longrightarrow
  \RHNI(\gamma,sep,f,br,locs,O,P).
\]
Safety of the layout and well-formedness of $P$ are deliberately not fields of
this definition, they are hypotheses of the preservation theorem.  This
matches the Lean definition \texttt{RobustRowhammerNI}.
\end{definition}

For convenience, the mechanisation also defines
\[
\begin{split}
  \Admissible(sep,f,br,locs,P)
  \quad\Longleftrightarrow\quad{}
  &\mathsf{HasProbFlipFrame}(sep,f)
  \\
  &{}\land\;sep.\mathsf{isSafe}(br,locs)
  \\
  &{}\land\;P.\mathsf{WellFormed}(locs).
\end{split}
\]
The checked theorem statements below take these three assumptions as
separate arguments, exactly as their Lean counterparts do. We can now
define the four concrete checked observers, promised above. Towards
that aim, let
\[
  L_{\gamma,locs}
  \triangleq
  \{x:\Loc\mid x\in locs\land\gamma(x)=\Low\}.
\]
The development defines the following protected observers.

\begin{itemize}

\item \textbf{Protected-low-memory observer.}  It hides residual
  syntax, terminal status, and the access trace:
  \[
  O^{\mathsf{mem}}_{\gamma,locs}(Q,\bar m,s)
  \triangleq
  \bar m|_{L_{\gamma,locs}}.
\]

\item \textbf{Low-access observer.}  Since labels contain addresses
  but not transferred values, this observer captures address-trace
  leakage rather than value-labelled I/O leakage.
\[
  O^{\mathsf{access}}_\gamma(Q,\bar m,s)
  \triangleq
  \LowAccessTrace_\gamma(s),
\]
where the filter retains $\mathsf{read}(x)$ and $\mathsf{write}(x)$ exactly
when $\gamma(x)=\Low$.

\item \textbf{Progress-sensitive observer.}  Because the definition compares every finite run, this
observer reveals exact abstract-machine progress, including the first step-count
at which termination is visible.
\[
\begin{split}
  O^{\mathsf{progress}}_{\gamma,locs}(Q,\bar m,s) \triangleq \bigl(
  \Terminal(Q), \bar m|_{L_{\gamma,locs}}, \LowAccessTrace_\gamma(s)
  \bigr),
\end{split}
\]
where $\Terminal(Q)$ is the Boolean that is true exactly for
$Q=\mathsf{skip}$.

\item \textbf{Terminal-status-hidden observer.}  This omits the
  terminal-status bit but is not claimed to formalise classical
  termination-insensitive non-interference: low memory or low accesses
  at fixed step count may still reveal progress.
  \[
  O^{\mathsf{hidden}}_{\gamma,locs}(Q,\bar m,s)
  \triangleq
  \bigl(\bar m|_{L_{\gamma,locs}},\LowAccessTrace_\gamma(s)\bigr).
  \]

\end{itemize}

\paragraph{Checked theorem statements.}

The following statements are proved in Lean without \texttt{sorry} and without
new problem-specific axioms.

\begin{theorem}[Protected semantic collapse]
\label{thm:ni-protected-collapse}
Let $f$ be frame-local, let $sep.\mathsf{isSafe}(br,locs)$ hold, and let
$P.\mathsf{WellFormed}(locs)$ hold.  Then, for every memory $m$ and
$n$,
\[
  \RHProtectedView{n}_{sep,f,br,locs}(P,m)
  =
  \dirac\bigl(\DProtectedView{n}_{locs}(P,m)\bigr).
\]
This is theorem \texttt{protected\_semantic\_collapse}, a direct restatement
of the existing theorem \texttt{physical\_separation\_sound}.
\end{theorem}

\begin{theorem}[Observation collapse]
\label{thm:ni-observation-collapse}
Under the same frame-locality, safety, and well-formedness assumptions, for
every protected observer $O$, memory $m$, and step count $n$,
\[
  \ObsRH{n}_{sep,f,br,locs,O}(P,m)
  =
  \dirac\bigl(\ObsOrd{n}_{locs,O}(P,m)\bigr).
\]
This is theorem \texttt{observation\_collapse}.
\end{theorem}

\begin{theorem}[Non-interference equivalence]
\label{thm:ni-equivalence}
Let $f$ be frame-local, let $sep.\mathsf{isSafe}(br,locs)$ hold, and let
$P.\mathsf{WellFormed}(locs)$ hold.  Then, for every security policy $\gamma$
and protected observer $O$,
\[
  \OrdNI(\gamma,locs,O,P)
  \quad\Longleftrightarrow\quad
  \RHNI(\gamma,sep,f,br,locs,O,P).
\]
This is theorem \texttt{noninterference\_equivalence}.
\end{theorem}

\begin{corollary}[Robust preservation]
\label{cor:ni-robust-preservation}
If $sep.\mathsf{isSafe}(br,locs)$,
$P.\mathsf{WellFormed}(locs)$, and \newline $\OrdNI(\gamma,locs,O,P)$ hold, then
\[
  \RobustRHNI(\gamma,sep,br,locs,O,P).
\]
Equivalently, every frame-local kernel $f$ satisfies
\[
  \RHNI(\gamma,sep,f,br,locs,O,P).
\]
This is theorem \texttt{robust\_noninterference\_preservation}.
\end{corollary}

\begin{corollary}[Reflection for a fixed kernel]
\label{cor:ni-reflection-fixed}
Under frame-locality, safety, and well-formedness,
\[
  \neg\RHNI(\gamma,sep,f,br,locs,O,P)
  \Longrightarrow
  \neg\OrdNI(\gamma,locs,O,P).
\]
This is theorem \texttt{rowhammer\_interference\_reflection}.
\end{corollary}

\begin{corollary}[Existential reflection]
\label{cor:ni-reflection-exists}
Assume $sep.\mathsf{isSafe}(br,locs)$ and
$P.\mathsf{WellFormed}(locs)$.  Then
\[
\begin{split}
  \bigl(\exists f.\;&
    \mathsf{HasProbFlipFrame}(sep,f)
    \land
    \neg\RHNI(\gamma,sep,f,br,locs,O,P)
  \bigr)
  \\
  &\Longrightarrow
  \neg\OrdNI(\gamma,locs,O,P).
\end{split}
\]
\end{corollary}

\NI In the Lean code, this is theorem
\texttt{exists\_rowhammer\_interference\_reflection}.

\subsection{Proof sketches}

\paragraph{Structural preservation.}
The existing lemma \texttt{stepDHat\_wf} proves that one absorbing
deterministic step preserves well-formedness over $locs$.  Its proof follows
the syntax of the residual program.  Expression reduction preserves
expression well-formedness. Sequencing preserves the well-formedness of the
untouched continuation. A conditional either reduces its guard or selects an
already well-formed branch. Unfolding a loop constructs a conditional and
sequence from already well-formed components.  This lemma is what allows the
step bound simulation invariant to be applied recursively to the
successor program.

\paragraph{One-step agreement.}
The generic relation $\Agree(keep,m_f,m_d)$ states that the faulty and
deterministic memories agree at every location satisfying $keep$.
Lemma \texttt{stepRHHat\_agree} assumes that every location used by
the program is kept and that every supported fault outcome preserves
all kept locations.  If the input memories agree on $keep$, then every
supported Rowhammer successor has exactly the same residual program as
the deterministic successor and its memory still agrees with the
deterministic memory on $keep$.  The proof is by structural induction
on the program, using corresponding agreement lemmas for arithmetic
and Boolean expressions.  In a read or write case, frame preservation
ensures that the sampled fault cannot alter a kept location.  Since
values read from kept locations agree, both semantics reduce the same
expression, write the same value, select the same branch, and unfold
the same control construct.

\paragraph{Finite step bound invariant.}
Lemma \texttt{runRHTrace\_protected} iterates one-step agreement by induction
on $n$.  For every outcome $r$ in the support of the $n$-step Rowhammer trace
run, it proves three facts:
\begin{itemize}

  \item The residual program of $r$ equals the deterministic residual program.
  \item The memories agree on every kept location.
  \item The complete accumulated label lists are equal.

\end{itemize}
In the successor case, \texttt{stepRHHat\_agree} supplies the next related
configuration, \texttt{stepDHat\_wf} supplies well-formedness of its residual
program, and the induction hypothesis supplies the remaining step bound.

\paragraph{Instantiation by physical separation.}
For protected semantic collapse, the proof instantiates
$keep(x)$ with $x\in locs$.  The premise that faults preserve kept locations
is derived by combining frame-locality with safe separation: for an access to
$x\in locs$, \texttt{safe\_disjoint} shows that every $\ell\in locs$ lies
outside the victim set, and \texttt{outside\_unchanged} therefore preserves
$\ell$ in every supported fault outcome.  The finite step bound invariant then
shows that every supported Rowhammer outcome has the same protected view as
the deterministic run.  The finite-distribution lemma
\texttt{mapD\_eq\_dirac\_of\_support} turns this support-wise fact into the
point mass equation of Theorem~\ref{thm:ni-protected-collapse}.

\paragraph{Observation collapse.}
Unfolding the two observation definitions exposes a map of $O$ over the
Rowhammer protected-view distribution.  Rewriting that distribution with
Theorem~\ref{thm:ni-protected-collapse} leaves $O$ mapped over a Dirac
distribution.  The finite-distribution law
\[
  \FDist.\mathsf{map}\;O\;(\dirac(v))
  =
  \dirac(O(v))
\]
gives Theorem~\ref{thm:ni-observation-collapse}.

\paragraph{Non-interference equivalence and its corollaries.}
For low-equivalent $m_1$ and $m_2$, observation collapse rewrites the two
Rowhammer observation distributions to
\[
  \dirac\bigl(\ObsOrd{n}_{locs,O}(P,m_1)\bigr)
  \quad\text{and}\quad
  \dirac\bigl(\ObsOrd{n}_{locs,O}(P,m_2)\bigr).
\]
Ordinary non-interference implies equality of the two points and hence
congruence of their Dirac distributions.  Conversely, equality of the two
Dirac distributions implies equality of their points by injectivity of
$\dirac$.  This proves Theorem~\ref{thm:ni-equivalence}.  Robust preservation
introduces an arbitrary frame-local kernel and applies the forward direction.
Fixed-kernel and existential reflection are direct contrapositives of that
same direction.

\subsection{No-fault sanity results}
The checked development also proves, without any safety or well-formedness
assumption, that the no-fault kernel embeds deterministic execution exactly:
\[
  \RunRHTrace{n}_{sep,\NoFlip(sep),br,locs}(c)
  =
  \dirac\bigl(\RunDTrace{n}(c)\bigr).
\]
The proof is an induction on $n$ using the existing one-step no-fault theorem,
Dirac bind, and mapping over a Dirac distribution.  It yields
\[
  \ObsRH{n}_{sep,\NoFlip(sep),br,locs,O}(P,m)
  =
  \dirac\bigl(\ObsOrd{n}_{locs,O}(P,m)\bigr)
\]
and therefore
\[
  \OrdNI(\gamma,locs,O,P)
  \quad\Longleftrightarrow\quad
  \RHNI(\gamma,sep,\NoFlip(sep),br,locs,O,P).
\]
In Lean they are \texttt{runRHTrace\_noFlip},
\texttt{noFault\_observation}, and
\texttt{ordinaryNI\_iff\_noFaultNI}.
 
\section{Conclusion}\label{section_conclusion}

This paper has given an abstract, compositional operational semantics
for a representative programming language in the presence of
Rowhammer.  The semantics treats Rowhammer as a probabilistic effect
on memory accesses, rather than as a direct model of DRAM organisation
or semiconductor physics.  The no-fault embedding theorem,
\Cref{thm:no-fault-finite-horizon}, recovers the ordinary
deterministic semantics as the point mass special case of the
Rowhammer semantics, answering \RQONE.  The model isolates the
Rowhammer effect as a small generalisation from state transformers to
probabilistic state transformers: memory reads and writes are
intercepted by a fault kernel, while the rest of the operational
semantics is lifted compositionally using the probability monad.  This
gives the small, conceptually non-intrusive interface sought in
\RQTWO.  The interface is also largely independent of the ambient
programming language.  The development uses only that programs perform
memory reads and writes, and that the ordinary operational semantics
can expose those accesses.  This is not by itself a formal semantics
for every language or machine model, but it gives a clear route for
instantiating the same Rowhammer interface in other languages,
intermediate representations, or instruction-set semantics.  In this
sense, the paper affirms \RQTHREE.  The semantic-collapse theorem,
\Cref{thm:goal-formal}, proves the soundness of physical separation:
if every protected program access has a victim set disjoint from the
protected locations, then every admissible Rowhammer kernel induces
exactly the deterministic protected behaviour.  This justifies
physical separation as a Rowhammer defence within the abstract fault
contract, answering \RQFOUR.  Finally, the information flow results
show that the abstraction supports standard semantic security
reasoning.  In particular, Theorems \ref{thm:ni-protected-collapse},
\ref{thm:ni-observation-collapse} and \ref{thm:ni-equivalence}
transport non-interference reasoning from ordinary executions to
Rowhammer-affected executions, and back again, under the same
physical separation hypotheses.  Thus the model does not merely
describe faulty executions, it makes Rowhammer amenable to established
semantic security frameworks, addressing \RQFIVE.
 \section{Future work}\label{section_future_work}

\paragraph{Refining the Rowhammer model}
The fault kernel contract history-freeness (F3) and stationarity (F4)
exclude genuinely history-dependent faults such as charge accumulation
across repeated activations and refresh-cycle effects.  Generalising
requires extending the Rowhammer state with an activation history,
refresh state, or other hardware state. Our theorems would then
quantify over arbitrary stateful kernels satisfying the same
protected-region preservation condition, and we expect the
parametric-invariant proof architecture to carry over.  Moreover, DRAM
vendors typically deploy proprietary, undocumented mitigations, most
notably Target Row Refresh (TRR), whose opaque nature complicates
formal security guarantees.  Bridging the gap between our formal
models and these black-box defences requires rigorous experimental
reverse engineering.  Pioneering methodologies such as the U-TRR
framework \cite{hassan2021utrr}, uncover the internal mechanics of TRR
by taking data-retention failures as a side channel.  Alongside this,
complementary characterisation of DRAM's disturbance sensitivities
\cite{orosa2021deeperlook} establishes clear empirical boundaries on
when such defences succeed or fail, and in turn informs principled
mitigations such as BlockHammer \cite{yaglikci2021blockhammer}.  By
experimentally exposing the true behaviour and limitations of in-DRAM
mitigations, such efforts provide the precise physical parameters and
fault models that a formal framework like ours needs in order to
deliver end-to-end software security guarantees.

\paragraph{Deploying physical separation.}
Our soundness theorem is conditional on the data placement satisfying
physical separation.  Establishing that on modern hardware is
nontrivial: DRAM vendors do not publish row geometry, in-DRAM address
remapping obscures physical adjacency, and the effective blast radius
must be measured per device.  Empirical work on learning the Rowhammer
parameters of a given device \cite{NaserediniA:alaactlorm} is
complementary to our results: it supplies the physical facts that
instantiate $br$ and $\safe$, while the theorem turns those facts into
a semantic guarantee.  The separation between the abstract geometry
(\texttt{SeparationModel}) and the semantics in our mechanisation is
designed exactly for this division of labour.

\paragraph{Probabilistic non-termination.}
A natural extension is to complement the finite step bound semantics
with an explicit account of eventual termination and divergence.  This
could be done either by taking limits of the finite step bound
termination probabilities, by developing a sub-probability semantics
over final memories, or by constructing a probability measure on
infinite execution paths.  A fully event-based treatment, in which
divergence is identified with a measurable set of infinite executions,
may additionally require the construction of an appropriate
probability measure on infinite paths and an adequacy theorem relating
that measure to the $n$-step operational semantics.  Each route
introduces nontrivial order-theoretic or measure-theoretic
mechanisation.

\paragraph{Connections with Shannon-Scott information.}
A further direction is to relate the present probabilistic account of
Rowhammer to semantic accounts of information based on both Shannon
entropy and Scott domains.  Shannon's theory of communication
provides a natural language for noisy channels and probabilistic
leakage~\cite{Shannon1948}, while domain-theoretic and topological
methods are commonly used in programming language semantics to model
partial information and non-terminating computation.  Recent work has
connected these two perspectives in the setting of information flow
and non-interference~\cite{Hunt2023ShannonScott} and Shannon's lattice
theory of information~\cite{shannon1953lattice}. Our results compare
the distributions of low observations after a fixed number of
execution steps.  It would be interesting to investigate whether this
account can be extended to an infinite-horizon semantics in which
Rowhammer-induced probabilistic leakage, partial observations, and
non-termination are treated in a single Shannon-Scott framework.
 {\small

  \section*{Acknowledgements}
  \NI Chunyan Mu drew our attention to Rowhammer being an open problem
  from a non-interference angle.  We thank David Clark for numerous
  discussions about information flow and non-interference, Davide
  Bartolini and Xubin Tan for helping us to understand memory
  controllers, and Lev Mukhanov for alerting us to the difficulty of
  reproducing the Rowhammer attacks.  Fredrik Dahlqvist helped us
  better to understand Kozen-style probabilistic semantics of
  probabilistic programming languages.  We thank Yusuke Izawa and
  Mohammad M.~Ahmadpanah for helpful comments, which greatly improved
  the presentation of this paper.

}
 
\clearpage
\bibliographystyle{ACM-Reference-Format} \bibliography{references}

\newpage
\appendix
\vspace{1em}

{\huge \NI Supplementary Appendix: Mechanised high-level operational semantics of Rowhammer}
 \begin{bibunit} \section{Material for \NoCaseChange{\Cref{section_language}}:
  A deterministic small-step operational semantics for a WHILE
  language}\label{sec:small-step-ordinary}

This appendix presents the full small-step operational semantics that
were given only partially in \Cref{section_language}.  For more on
small-step operational semantics see any textbook on programming
language semantics like \EG
\cite{NipkowKlein2014,RiisNielsonH:semwitaaa,WinskelG:forsemoplai}.

Recall that the syntax (arithmetic expressions, Boolean expressions,
programs) are given by
\begin{GRAMMAR}
  e &::=& c \mid x \mid e+e
  \\[2mm]
  b &::=&
    \mathsf{true}
    \mid \mathsf{false}
    \mid \neg b
    \mid b\wedge b
    \mid e=e
  \\[2mm]
  P &::=&
    \mathsf{skip}
    \mid x:=e
    \mid P; P
    \mid \mathsf{if}\ b\ \mathsf{then}\ P\ \mathsf{else}\ P
    \mid \mathsf{while}\ b\ \mathsf{do}\ P
\end{GRAMMAR}
We identify arithmetic constants with their denoted values.  Thus
$v,v_1,v_2\in\Val$ range over arithmetic values, represented by constants
in the expression syntax.  Boolean values are
$t,t_1,t_2\in\{\mathsf{true},\mathsf{false}\}$.

Evaluation is left-to-right.  Conjunction is deliberately
non-short-circuiting, matching the expression semantics used elsewhere in
the paper.  Although deterministic expression reduction does not modify
memory, memories are included in expression configurations so that the
ordinary and faulty transition systems have the same shape.

Transitions carry a \emph{label} recording the memory access they perform:
\[
  \ell ::= \lread(x) \mid \lwrite(x) \mid \tau .
\]
A read of variable $x$ is labelled $\lread(x)$, the write performed by an
assignment to $x$ is labelled $\lwrite(x)$, and every other transition is
silent ($\tau$); congruence rules propagate the label of their premise.
The labels play no role in the deterministic semantics itself, but they
localise the fault points of the probabilistic Rowhammer semantics of
\Cref{section_probabilistic_model_of_rowhammer} and make access counting definitional.
Single steps are always written with their label; every multi-step notion
below is indexed by the \emph{sequence} of labels performed (written with
a double arrow), so the labelled judgement is the only transition
relation in play.

\subsubsection{Arithmetic expressions}

An arithmetic-expression configuration has the form $\langle e,m\rangle$.
The labelled relation
\[
  \langle e,m\rangle \AStepL{\ell} \langle e',m'\rangle
\]
is the least relation generated by the following rules.  Arithmetic values
have no outgoing transition.

\[
\frac{\strut}{
  \langle x,m\rangle
  \AStepL{\lread(x)}
  \langle m(x),m\rangle
}
\tag{\textsc{A-Read}}
\]

\[
\frac{
  \langle e_1,m\rangle
  \AStepL{\ell}
  \langle e_1',m'\rangle
}{
  \langle e_1+e_2,m\rangle
  \AStepL{\ell}
  \langle e_1'+e_2,m'\rangle
}
\tag{\textsc{A-Plus-L}}
\]

\[
\frac{
  \langle e_2,m\rangle
  \AStepL{\ell}
  \langle e_2',m'\rangle
}{
  \langle v_1+e_2,m\rangle
  \AStepL{\ell}
  \langle v_1+e_2',m'\rangle
}
\tag{\textsc{A-Plus-R}}
\]

\[
\frac{
  v=v_1+v_2
}{
  \langle v_1+v_2,m\rangle
  \AStepL{\tau}
  \langle v,m\rangle
}
\tag{\textsc{A-Plus}}
\]

\subsubsection{Boolean expressions}

A Boolean-expression configuration has the form $\langle b,m\rangle$.
The labelled relation
\[
  \langle b,m\rangle \BStepL{\ell} \langle b',m'\rangle
\]
is generated by the following rules.  The Boolean values
$\mathsf{true}$ and $\mathsf{false}$ have no outgoing transition.

\[
\frac{
  \langle b,m\rangle
  \BStepL{\ell}
  \langle b',m'\rangle
}{
  \langle \neg b,m\rangle
  \BStepL{\ell}
  \langle \neg b',m'\rangle
}
\tag{\textsc{B-Not-Step}}
\]

\[
\frac{\strut}{
  \langle \neg\mathsf{true},m\rangle
  \BStepL{\tau}
  \langle \mathsf{false},m\rangle
}
\tag{\textsc{B-Not-True}}
\]

\[
\frac{\strut}{
  \langle \neg\mathsf{false},m\rangle
  \BStepL{\tau}
  \langle \mathsf{true},m\rangle
}
\tag{\textsc{B-Not-False}}
\]

\[
\frac{
  \langle b_1,m\rangle
  \BStepL{\ell}
  \langle b_1',m'\rangle
}{
  \langle b_1\wedge b_2,m\rangle
  \BStepL{\ell}
  \langle b_1'\wedge b_2,m'\rangle
}
\tag{\textsc{B-And-L}}
\]

\[
\frac{
  \langle b_2,m\rangle
  \BStepL{\ell}
  \langle b_2',m'\rangle
}{
  \langle t_1\wedge b_2,m\rangle
  \BStepL{\ell}
  \langle t_1\wedge b_2',m'\rangle
}
\tag{\textsc{B-And-R}}
\]

\[
\frac{
  t=t_1\wedge t_2
}{
  \langle t_1\wedge t_2,m\rangle
  \BStepL{\tau}
  \langle t,m\rangle
}
\tag{\textsc{B-And}}
\]

\[
\frac{
  \langle e_1,m\rangle
  \AStepL{\ell}
  \langle e_1',m'\rangle
}{
  \langle e_1=e_2,m\rangle
  \BStepL{\ell}
  \langle e_1'=e_2,m'\rangle
}
\tag{\textsc{B-Eq-L}}
\]

\[
\frac{
  \langle e_2,m\rangle
  \AStepL{\ell}
  \langle e_2',m'\rangle
}{
  \langle v_1=e_2,m\rangle
  \BStepL{\ell}
  \langle v_1=e_2',m'\rangle
}
\tag{\textsc{B-Eq-R}}
\]

\[
\frac{
  t=(v_1=v_2)
}{
  \langle v_1=v_2,m\rangle
  \BStepL{\tau}
  \langle t,m\rangle
}
\tag{\textsc{B-Eq}}
\]

\subsubsection{Programs}

A program configuration has the form $\langle P,m\rangle$.  The labelled
program transition relation
\[
  \langle P,m\rangle \PStepL{\ell} \langle P',m'\rangle
\]
is the least relation generated by the following rules.

An assignment first reduces its right-hand side one expression step at a
time.  Once the expression is a value, the assignment performs the write.

\[
\frac{
  \langle e,m\rangle
  \AStepL{\ell}
  \langle e',m'\rangle
}{
  \langle x:=e,m\rangle
  \PStepL{\ell}
  \langle x:=e',m'\rangle
}
\tag{\textsc{P-Assign-Step}}
\]

\[
\frac{\strut}{
  \langle x:=v,m\rangle
  \PStepL{\lwrite(x)}
  \langle \mathsf{skip},m[x\mapsto v]\rangle
}
\tag{\textsc{P-Assign}}
\]

Sequential composition reduces its left-hand program until that program
has terminated.

\[
\frac{
  \langle P,m\rangle
  \PStepL{\ell}
  \langle P',m'\rangle
}{
  \langle P;Q,m\rangle
  \PStepL{\ell}
  \langle P';Q,m'\rangle
}
\tag{\textsc{P-Seq-Step}}
\]

\[
\frac{\strut}{
  \langle \mathsf{skip};Q,m\rangle
  \PStepL{\tau}
  \langle Q,m\rangle
}
\tag{\textsc{P-Seq-Skip}}
\]

A conditional reduces its guard one Boolean-expression step at a time and
then selects the corresponding branch.

\[
\frac{
  \langle b,m\rangle
  \BStepL{\ell}
  \langle b',m'\rangle
}{
  \left\langle
    \mathsf{if}\ b\ \mathsf{then}\ P\ \mathsf{else}\ Q,
    m
  \right\rangle
  \PStepL{\ell}
  \left\langle
    \mathsf{if}\ b'\ \mathsf{then}\ P\ \mathsf{else}\ Q,
    m'
  \right\rangle
}
\tag{\textsc{P-If-Step}}
\]

\[
\frac{\strut}{
  \left\langle
    \mathsf{if}\ \mathsf{true}\ \mathsf{then}\ P\ \mathsf{else}\ Q,
    m
  \right\rangle
  \PStepL{\tau}
  \langle P,m\rangle
}
\tag{\textsc{P-If-True}}
\]

\[
\frac{\strut}{
  \left\langle
    \mathsf{if}\ \mathsf{false}\ \mathsf{then}\ P\ \mathsf{else}\ Q,
    m
  \right\rangle
  \PStepL{\tau}
  \langle Q,m\rangle
}
\tag{\textsc{P-If-False}}
\]

A while loop unfolds by one standard small step.  Its guard is subsequently
reduced by the conditional rules above.

\[
\frac{\strut}{
  \left\langle
    \mathsf{while}\ b\ \mathsf{do}\ P,
    m
  \right\rangle
  \PStepL{\tau}
  \left\langle
    \mathsf{if}\ b\ \mathsf{then}\
      \bigl(P;\mathsf{while}\ b\ \mathsf{do}\ P\bigr)\
    \mathsf{else}\ \mathsf{skip},
    m
  \right\rangle
}
\tag{\textsc{P-While}}
\]

There is no transition from $\langle\mathsf{skip},m\rangle$; these are the
terminal program configurations.

\subsubsection{Finite and terminating executions}

A finite execution is indexed by the sequence of labels it performs.
Write
\[
  \langle P,m\rangle
  \PStepT{s}
  \langle P',m'\rangle
\]
for $s=\ell_1\cdots\ell_k$ when there are configurations
\[
  \langle P,m\rangle
  = C_0 \PStepL{\ell_1} C_1 \PStepL{\ell_2} \cdots
    \PStepL{\ell_k} C_k
  = \langle P',m'\rangle.
\]
The length $|s|$ is the number of transitions performed. Because
expression reductions are propagated through assignments and conditionals
one step at a time, it counts both expression reductions and command
reductions.  The \emph{access trace} $\mathsf{acc}(s)$ is the subsequence
of non-$\tau$ labels of $s$, in order; its length is the number of memory
accesses (reads and writes) the run performs.  Both are definitions over
the label sequence, not counts of rule occurrences inside derivation
trees.  The analogous trace-indexed relations $\AStepT{s}$ and
$\BStepT{s}$ at the expression levels are defined in the same way.

A program terminates with final memory $m'$, written
\[
  \langle P,m\rangle\Downarrow m',
\]
when $\langle P,m\rangle\PStepT{s}\langle\mathsf{skip},m'\rangle$ for some
label sequence $s$.  It terminates within at most $n$ small steps, written
\[
  \langle P,m\rangle\Downarrow_{\leq n}m',
\]
when moreover some such $s$ has $|s|\leq n$.  A program diverges from $m$,
written
\[
  \langle P,m\rangle\Uparrow,
\]
when for every $n\in\Nat$ there are a label sequence $s$ with $|s|=n$ and
a configuration $\langle P_n,m_n\rangle$ such that
$\langle P,m\rangle\PStepT{s}\langle P_n,m_n\rangle$.
The index $n$ here is a
literal bound on the number of small-step transitions.  Sequential
composition cannot duplicate or reuse a budget already consumed by its
left-hand program.

\subsubsection{Basic metatheory}

The lemmas of this subsection are standard and proved on paper. The
mechanisation contains those the soundness theorem rests on
(determinism and progress in functional form, trace realisation,
divergence, and well-formedness preservation).
The remaining lemmas below are stated for completeness of the
presentation and are not part of the Lean development.

\begin{lemma}[Values are terminal]
\label{lem:ss-values-terminal}
There are no $\ell$ and $C$ with $\langle v,m\rangle\AStepL{\ell}C$; none
with $\langle t,m\rangle\BStepL{\ell}C$; and none with
$\langle\mathsf{skip},m\rangle\PStepL{\ell}C$.
\end{lemma}

\begin{lemma}[Determinism]
\label{lem:ss-determinism}
Each of the labelled relations is a partial function on configurations,
including its label: if $C\PStepL{\ell_1} C_1$ and $C\PStepL{\ell_2} C_2$
then $\ell_1=\ell_2$ and $C_1=C_2$, and likewise for $\AStepL{\cdot}$ and
$\BStepL{\cdot}$.
\end{lemma}

\begin{proof}
Rule induction.  For each syntactic form at most one rule applies: the
congruence rules require a non-value in the position being reduced (values
are terminal by \Cref{lem:ss-values-terminal}, so their premises are
otherwise unsatisfiable), and the redex rules require values in all
positions.
\end{proof}

\begin{lemma}[Progress]
\label{lem:ss-progress}
If $e$ is not a value then $\langle e,m\rangle\AStepL{\ell}\langle
e',m'\rangle$ for some $\ell$, $e'$, $m'$; if $b$ is not a value then
$\langle b,m\rangle\BStepL{\ell}\langle b',m'\rangle$ for some $\ell$,
$b'$, $m'$; and if $P\neq\mathsf{skip}$ then
$\langle P,m\rangle\PStepL{\ell}\langle P',m'\rangle$ for some $\ell$,
$P'$, $m'$.
\end{lemma}

\begin{proof}
Induction on the syntax, using totality of memories for \textsc{A-Read}.
\end{proof}

\begin{corollary}[Termination--divergence dichotomy]
\label{cor:ss-dichotomy}
For every configuration $\langle P,m\rangle$, exactly one of the following
holds: $\langle P,m\rangle\Downarrow m'$ for a unique $m'$, or
$\langle P,m\rangle\Uparrow$.
\end{corollary}

\begin{proof}
By \Cref{lem:ss-determinism} the maximal transition sequence from
$\langle P,m\rangle$ is unique, and by \Cref{lem:ss-progress} it can stop
only at a terminal configuration $\langle\mathsf{skip},m'\rangle$.  If that
sequence is finite it yields $\Downarrow m'$, with $m'$ unique by
determinism iterated along the run; if it is infinite, every $n$ is
realised and $\Uparrow$ holds.  The two cases exclude each other because a
terminated run has no continuation.
\end{proof}

\begin{lemma}[Memory invariance of deterministic expression steps]
\label{lem:ss-mem-invariance}
If $\langle e,m\rangle\AStepL{\ell}\langle e',m'\rangle$ or
$\langle b,m\rangle\BStepL{\ell}\langle b',m'\rangle$, then $m'=m$.
\end{lemma}

Deterministically, expression-level small steps are thus redundant: only
\textsc{P-Assign} modifies the memory.  This lemma is the precise sense in
which the fine-grained presentation is overkill for the ordinary semantics;
the Rowhammer refinement of \Cref{sec:small-step-faulty} breaks exactly
this lemma, a read may disturb the blast radius of the accessed
location, and nothing else in the design.

\begin{lemma}[Expression evaluation terminates, with explicit cost]
\label{lem:ss-expr-cost}
Let $r(e)$ be the number of variable occurrences plus the number of
operators in $e$.  Then $\langle e,m\rangle\AStepT{s}\langle v,m\rangle$
for a unique value $v$ and some $s$ with $|s|=r(e)$, and similarly for
Boolean expressions.  Each transition fires exactly one \textsc{A-Read}
or one redex rule under congruences, so $r$ strictly decreases;
$\mathsf{acc}(s)$ is the sequence of reads of $e$, in left-to-right
order.
\end{lemma}

\begin{lemma}[Adequacy of explicit reads]
\label{lem:ss-adequacy}
$\langle e,m\rangle\AStepT{s}\langle v,m\rangle$ for some $s$ iff
$\EvalOrd{e}(m)=v$, and $\langle b,m\rangle\BStepT{s}\langle t,m\rangle$
for some $s$ iff $\EvalOrd{b}(m)=t$, where $\EvalOrd{\cdot}$ is the
atomic expression evaluation defined above.
\end{lemma}

Each dynamic variable read is exactly one $\lread(x)$-labelled transition
and each dynamic assignment write exactly one $\lwrite(x)$-labelled
transition, so the number of potential access-triggered Rowhammer events
in a run is the length of its access trace.  The total transition count is
an exact cost for this abstract machine, but it is \emph{not} the number
of memory accesses: it also counts operator reductions, branch selection,
sequencing elimination, and loop unfolding, silent ($\tau$) steps that
trigger no fault in the semantics of \Cref{section_probabilistic_model_of_rowhammer}.
 \section{Material for \NoCaseChange{\Cref{section_probabilistic_model_of_rowhammer}}:
A probabilistic model of Rowhammer}\label{sec:small-step-faulty}

This appendix presents the full small-step probabilistic operational
semantics that was given only partially in
\Cref{section_probabilistic_model_of_rowhammer}.  For more on
probabilistic operational semantics see \EG
\cite{DahlqvistF:sempropagi,DiPierroetal2010,LagoZorzi2012}.

The syntax and evaluation order are those of
\Cref{sec:small-step-ordinary}.  Throughout, the blast radius $br$ and
the protected set $locs$ are ambient parameters, and $\pflip$ is a
fault kernel subject to the contract (F1--F4) of
\Cref{sec:fault-kernel-contract}.  The judgements
\[
  \langle e,m\rangle \RHAStep \mu,
  \qquad
  \langle b,m\rangle \RHBStep \mu,
  \qquad
  \langle P,m\rangle \RHPStep \mu
\]
relate a configuration to a \emph{distribution} over successor
configurations, with
$\mu\in\FDist(\mathsf{AExp}\times\Mem)$,
$\mu\in\FDist(\mathsf{BExp}\times\Mem)$, and
$\mu\in\FDist(\mathsf{Prog}\times\Mem)$ respectively.  Exactly two
rules sample the kernel, \textsc{RH-A-Read} and \textsc{RH-P-Assign}:
every redex rule is the Dirac lifting of its deterministic
counterpart, and every congruence rule pushes the reduction context
through the premise distribution with $\FDist.map$.

\subsubsection{Arithmetic expressions}

Arithmetic values have no outgoing transition.  A variable read first
obtains the value $m(x)$ and then samples the fault triggered by that
access.

\[
\frac{\strut}{
  \langle x,m\rangle
  \RHAStep
  \left(
    \mathbf{let}\ m'\leftarrow
      \pflip\bigl(\adj(br,locs,x),m\bigr)
    \ \mathbf{in}\
    \dirac\bigl(\langle m(x),m'\rangle\bigr)
  \right)
}
\tag{\textsc{RH-A-Read}}
\]

\[
\frac{
  \langle e_1,m\rangle \RHAStep \mu
}{
  \langle e_1+e_2,m\rangle
  \RHAStep
  \FDist.map
    \bigl(
      \lambda\langle e_1',m'\rangle.\,
        \langle e_1'+e_2,m'\rangle
    \bigr)
    (\mu)
}
\tag{\textsc{RH-A-Plus-L}}
\]

\[
\frac{
  \langle e_2,m\rangle \RHAStep \mu
}{
  \langle v_1+e_2,m\rangle
  \RHAStep
  \FDist.map
    \bigl(
      \lambda\langle e_2',m'\rangle.\,
        \langle v_1+e_2',m'\rangle
    \bigr)
    (\mu)
}
\tag{\textsc{RH-A-Plus-R}}
\]

\[
\frac{
  v=v_1+v_2
}{
  \langle v_1+v_2,m\rangle
  \RHAStep
  \dirac\bigl(\langle v,m\rangle\bigr)
}
\tag{\textsc{RH-A-Plus}}
\]

\subsubsection{Boolean expressions}

The Boolean values $\mathsf{true}$ and $\mathsf{false}$ have no
outgoing transition.

\[
\frac{
  \langle b,m\rangle \RHBStep \mu
}{
  \langle\neg b,m\rangle
  \RHBStep
  \FDist.map
    \bigl(
      \lambda\langle b',m'\rangle.\,
        \langle\neg b',m'\rangle
    \bigr)
    (\mu)
}
\tag{\textsc{RH-B-Not-Step}}
\]

\[
\frac{\strut}{
  \langle\neg\mathsf{true},m\rangle
  \RHBStep
  \dirac\bigl(\langle\mathsf{false},m\rangle\bigr)
}
\tag{\textsc{RH-B-Not-True}}
\]

\[
\frac{\strut}{
  \langle\neg\mathsf{false},m\rangle
  \RHBStep
  \dirac\bigl(\langle\mathsf{true},m\rangle\bigr)
}
\tag{\textsc{RH-B-Not-False}}
\]

\[
\frac{
  \langle b_1,m\rangle \RHBStep \mu
}{
  \langle b_1\wedge b_2,m\rangle
  \RHBStep
  \FDist.map
    \bigl(
      \lambda\langle b_1',m'\rangle.\,
        \langle b_1'\wedge b_2,m'\rangle
    \bigr)
    (\mu)
}
\tag{\textsc{RH-B-And-L}}
\]

\[
\frac{
  \langle b_2,m\rangle \RHBStep \mu
}{
  \langle t_1\wedge b_2,m\rangle
  \RHBStep
  \FDist.map
    \bigl(
      \lambda\langle b_2',m'\rangle.\,
        \langle t_1\wedge b_2',m'\rangle
    \bigr)
    (\mu)
}
\tag{\textsc{RH-B-And-R}}
\]

\[
\frac{
  t=t_1\wedge t_2
}{
  \langle t_1\wedge t_2,m\rangle
  \RHBStep
  \dirac\bigl(\langle t,m\rangle\bigr)
}
\tag{\textsc{RH-B-And}}
\]

\[
\frac{
  \langle e_1,m\rangle \RHAStep \mu
}{
  \langle e_1=e_2,m\rangle
  \RHBStep
  \FDist.map
    \bigl(
      \lambda\langle e_1',m'\rangle.\,
        \langle e_1'=e_2,m'\rangle
    \bigr)
    (\mu)
}
\tag{\textsc{RH-B-Eq-L}}
\]

\[
\frac{
  \langle e_2,m\rangle \RHAStep \mu
}{
  \langle v_1=e_2,m\rangle
  \RHBStep
  \FDist.map
    \bigl(
      \lambda\langle e_2',m'\rangle.\,
        \langle v_1=e_2',m'\rangle
    \bigr)
    (\mu)
}
\tag{\textsc{RH-B-Eq-R}}
\]

\[
\frac{
  t=(v_1=v_2)
}{
  \langle v_1=v_2,m\rangle
  \RHBStep
  \dirac\bigl(\langle t,m\rangle\bigr)
}
\tag{\textsc{RH-B-Eq}}
\]

\subsubsection{Programs}

An assignment first reduces its right-hand side one expression step at
a time.  Once the expression is a value, the assignment writes that
value and samples the fault triggered by the write, on the
\emph{updated} memory.

\[
\frac{
  \langle e,m\rangle \RHAStep \mu
}{
  \langle x:=e,m\rangle
  \RHPStep
  \FDist.map
    \bigl(
      \lambda\langle e',m'\rangle.\,
        \langle x:=e',m'\rangle
    \bigr)
    (\mu)
}
\tag{\textsc{RH-P-Assign-Step}}
\]

\[
\frac{\strut}{
  \langle x:=v,m\rangle
  \RHPStep
  \left(
    \mathbf{let}\ m'\leftarrow
      \pflip\bigl(
        \adj(br,locs,x),
        m[x\mapsto v]
      \bigr)
    \ \mathbf{in}\
    \dirac\bigl(\langle\mathsf{skip},m'\rangle\bigr)
  \right)
}
\tag{\textsc{RH-P-Assign}}
\]

Sequential composition reduces its left-hand program until that
program has terminated.

\[
\frac{
  \langle P,m\rangle \RHPStep \mu
}{
  \langle P;Q,m\rangle
  \RHPStep
  \FDist.map
    \bigl(
      \lambda\langle P',m'\rangle.\,
        \langle P';Q,m'\rangle
    \bigr)
    (\mu)
}
\tag{\textsc{RH-P-Seq-Step}}
\]

\[
\frac{\strut}{
  \langle\mathsf{skip};Q,m\rangle
  \RHPStep
  \dirac\bigl(\langle Q,m\rangle\bigr)
}
\tag{\textsc{RH-P-Seq-Skip}}
\]

A conditional reduces its guard one Boolean-expression step at a time
and then selects the corresponding branch.

\[
\frac{
  \langle b,m\rangle \RHBStep \mu
}{
  \left\langle
    \mathsf{if}\ b\ \mathsf{then}\ P\ \mathsf{else}\ Q,
    m
  \right\rangle
  \RHPStep
  \FDist.map
    \left(
      \lambda\langle b',m'\rangle.\,
      \left\langle
        \mathsf{if}\ b'\ \mathsf{then}\ P\ \mathsf{else}\ Q,
        m'
      \right\rangle
    \right)
    (\mu)
}
\tag{\textsc{RH-P-If-Step}}
\]

\[
\frac{\strut}{
  \left\langle
    \mathsf{if}\ \mathsf{true}\ \mathsf{then}\ P\ \mathsf{else}\ Q,
    m
  \right\rangle
  \RHPStep
  \dirac\bigl(\langle P,m\rangle\bigr)
}
\tag{\textsc{RH-P-If-True}}
\]

\[
\frac{\strut}{
  \left\langle
    \mathsf{if}\ \mathsf{false}\ \mathsf{then}\ P\ \mathsf{else}\ Q,
    m
  \right\rangle
  \RHPStep
  \dirac\bigl(\langle Q,m\rangle\bigr)
}
\tag{\textsc{RH-P-If-False}}
\]

A while loop unfolds by one standard small step.  Its guard is
subsequently reduced by the conditional rules above.

\[
\frac{\strut}{
  \left\langle
    \mathsf{while}\ b\ \mathsf{do}\ P,
    m
  \right\rangle
  \RHPStep
  \dirac\left(
    \left\langle
      \mathsf{if}\ b\ \mathsf{then}\
        \bigl(P;\mathsf{while}\ b\ \mathsf{do}\ P\bigr)\
      \mathsf{else}\ \mathsf{skip},
      m
    \right\rangle
  \right)
}
\tag{\textsc{RH-P-While}}
\]

There is no transition from $\langle\mathsf{skip},m\rangle$. These are
the terminal program configurations.  The rules are syntax-directed,
so every non-terminal configuration determines a unique one-step
distribution, written $\RHStep(\langle P,m\rangle)$ at the program
level. The absorbing kernel $\RHStepHat$ and the finite runs
$\RunRH{n}$ and $\RunRHT{n}$ built from it are defined in
\Cref{sec:finite-prob-executions}.

\subsection{The mechanised proof of the central theorem}
\label{sec:mechanised-proof}

For reference we reproduce the Lean proof of \Cref{thm:goal-formal}
from \texttt{Soundness.lean}. The prose sketch is in
\Cref{sec:proof-sketch}.

\begin{lstlisting}[language=Lean]
theorem physical_separation_sound : PhysicalSeparationSound := by
  intro sep f hframe br locs hsafe p hwf m n
  apply Dist.mapD_eq_dirac_of_support
  intro r hr
  have hsub : ∀ x ∈ locs, x ∈ locs := fun x hx => hx
  have hpres : ∀ x ∈ locs, ∀ (m₀ m' : Memory),
      m' ∈ (f (sep.adj br locs x) m₀).support → ∀ ℓ, ℓ ∈ locs → m' ℓ = m₀ ℓ :=
    fun x hx m₀ m' hm' ℓ hℓ =>
      hframe.outside_unchanged _ m₀ m' ℓ hm' (sep.safe_disjoint hsafe hx ℓ hℓ)
  obtain ⟨h1, h2, h3⟩ :=
    runRHTrace_protected (keep := (· ∈ locs)) hsub hpres n hwf
      (AgreeP.refl _ m) r hr
  simp only [protView, Prod.mk.injEq]
  refine ⟨h1, ?_, ?_⟩
  · funext x
    exact h2 x.1 x.2
  · rw [h3]
\end{lstlisting}
 
\noindent
The script follows the three moves of the sketch.  The \texttt{intro}
line introduces the separation geometry, the kernel with its
frame-locality hypothesis (\texttt{hframe}, clause F2), the blast radius
and protected set, physical separation (\texttt{hsafe}), the well-formed
program (\texttt{hwf}), the initial memory, and the step bound.
\texttt{Dist.mapD\_eq\_dirac\_of\_support} is the support
characterisation of the point mass
(\Cref{section_mathematical_preliminaries}): it reduces the
distributional equation to showing that every outcome \texttt{r} in the
support of the trace-carrying Rowhammer run has the protected view of
the deterministic run.  The two \texttt{have}s discharge the two clauses
of the parametric invariant at $\keep\triangleq(\cdot\in locs)$:
\texttt{hsub} is clause (i), used locations are kept, trivial at this
instantiation, where it is exactly well-formedness, \texttt{hpres} is
clause (ii), fault samples preserve kept locations, composed from
\texttt{safe\_disjoint} (the victim set of a protected access avoids
$locs$) and \texttt{hframe.outside\_unchanged} (samples fix everything
outside the victim set).  The iterated invariant
\texttt{runRHTrace\_protected}, started from the reflexive agreement on
the initial memory, then yields the three components of the protected
view: \texttt{h1}, equal residual programs, \texttt{h2}, memories
agreeing on $locs$, \texttt{h3}, equal label sequences.  The remainder
repackages them as an equality of protected views: \texttt{funext}
turns pointwise agreement on protected locations into equality of the
restricted memories as functions on the subtype, and \texttt{h3}
rewrites the access traces.
 
\renewcommand{\refname}{Appendix References} \putbib                                      \end{bibunit}                                

\end{document}